\documentclass{scrartcl}

\usepackage{amsmath}
\usepackage{amsfonts}
\usepackage{amssymb}

\usepackage{amsthm}
\newtheorem{definition}{Definition}
\newtheorem{proposition}{Proposition}

\usepackage{tikz}
\usetikzlibrary{cd, shapes.geometric, decorations.markings}
\tikzset{->-/.style={decoration={markings, mark=at position .5 with {\arrow{>}}}, postaction={decorate}}}

\usepackage{subdepth}
\usepackage{microtype}

\usepackage{hyperref}

\newcommand\diset[2]{\binom{#1}{#2}}
\newcommand\id{\mathrm{id}}
\newcommand\op{\mathrm{op}}
\newcommand\Int{\mathbf{Int}} 
\newcommand\OG{\mathbf{OG}} 

\newcommand\Lens{\mathfrak L}
\newcommand\Set{\mathbf{Set}}
\newcommand\DCPO{\mathbf{DCPO}}

\newcommand\G{\mathcal G} 
\renewcommand\H{\mathcal H} 
\newcommand\D{\mathcal D} 
\newcommand\C{\mathcal C} 
\newcommand\R{\mathbb R} 

\title{The game semantics of game theory}
\author{Jules Hedges}
\date{}

\begin{document}

\maketitle

\begin{abstract}
	We use a reformulation of compositional game theory to reunite game theory with game semantics, by viewing an open game as the System and its choice of contexts as the Environment. Specifically, the system is jointly controlled by $n \geq 0$ noncooperative players, each independently optimising a real-valued payoff. The goal of the system is to play a Nash equilibrium, and the goal of the environment is to prevent it. The key to this is the realisation that lenses (from functional programming) form a dialectica category, which have an existing game-semantic interpretation.
	
	In the second half of this paper, we apply these ideas to build a compact closed category of `computable open games' by replacing the underlying dialectica category with a wave-style geometry of interaction category, specifically the Int-construction applied to the traced cartesian category of directed-complete partial orders.
\end{abstract}

\section{Introduction}

Although the mathematics of games shares a common ancestor in Zermelo's work on backward induction, it split early on into two subjects that are essentially disjoint: game theory and game semantics.
Game theory is the applied study of modelling real-world interacting agents, for example in economics or artificial intelligence.
Game semantics, by contrast, uses agents to model -- this time in the sense of \emph{semantics} rather than \emph{mathematical modelling} -- situations in which a system interacts with an environment, but neither would usually be thought of as agents in a philosophical sense.
On a technical level, game semantics not only restricts to the two-player zero-sum case, but moreover promotes one of the players to be \emph{the} Player, and demotes the other to mere Opponent.
This induces a deep logical duality that pervades game semantics, apparently destroying any hope of bridging the gap to game theory, which typically involves $n$ players treated symmetrically.

Compositional game theory \cite{hedges_towards_compositional_game_theory,hedges_etal_compositional_game_theory}, as its name suggests, is an attempt to introduce the principle of compositionality into game theory, motivated by practical concerns about modelling large (for example economic) systems.
It is loosely inspired by game semantics, as well as categorical quantum mechanics \cite{abramsky04,coecke_kissinger_picturing_quantum_processes} and much recent work in applied category theory (e.g. \cite{fong_algebra_open_interconnected_systems}).
On a technical level, game semantics involves (typically monoidal) categories in which games are the objects and strategies (with various conditions) are the morphisms, whereas open games form the \emph{morphisms} of a monoidal category.
This means that open games can be denoted by string diagrams, which is invaluable for working with them in practice.
As with other categories of open systems, ordinary ``closed'' games are recovered as \emph{scalars}, or endomorphisms of the monoidal unit (see \cite{abramsky05}), and depicted as string diagrams with trivial boundary.

Central to understanding open games is the concept of a \emph{context}, which is a compressed representation of a game-theoretic situation in which an open game can be played.
Whereas an ordinary game has a set of strategy profiles and a subset of those which are Nash equilibria, in an open game the equilibria depend on the context.
This is the key to reuniting game theory and game semantics: we ignore the linguistic coincidence of the term \emph{player}, and instead view an open game as the System and the choice of contexts as the Environment.

The essence of this idea is already contained in the following quote from the introduction of \cite{abramsky_semantics_interaction}: ``If Tom, Tim and Tony converse in a room, then from Tom's point of view, he is the System, and Tim and Tony form the Environment; while from Tim's point of view, he is the System, and Tom and Tony form the Environment.''
The view of open games presented in this paper makes this precise when Tom, Tim and Tony are players in a noncooperative game.

In \cite{hedges_morphisms_open_games} open games were reformulated in terms of \emph{lenses} from functional programming \cite{pickering_gibbons_wu_profunctor_optics}.
This was extremely useful as a technical trick, but lenses are  usually used as destructive update operators on data structures and it is unclear what they have to do with game theory, if anything.
The key was a comment by Dusko Pavlovic to the author that the category of lenses $\Lens$ is a dialectica category \cite{depaiva_dialectica_categories_report}; combined with a game-\emph{semantic} view of dialectica categories \cite{blass91} we can see open games in their true form: as an interleaving of game theory and game semantics.

Specifically we find that an open game is a dialogue of a particular sort played between a system and its environment.
The system is jointly controlled by $n \geq 0$ noncooperative players, each independently optimising a real-valued payoff.
The winning condition turns out to be Nash equilibrium: the goal of the system is to play an equilibrium, and the goal of the environment is to prevent it.
Specifically, an open game consists of three pieces of data: a set $\Sigma$ of strategy profiles, a labelling function $\Sigma \to \{ P\text{-strategies} \}$, and a winning (for $P$) relation $\mathbf E \subseteq \Sigma \times \{ O\text{-strategies} \}$.

Taking a step back, this is a rare example of a cross-link in the family tree of the mathematics of games.
From the common ancestor in Zermelo's theorem \cite{schwalbe_walker_zermelo_early_history_game_theory} there was an almost immediate split, with little contact or commonality between the branches.
One branch led to game theory via \cite{von_neumann_morgenstern_theory_games_economic_behaviour} and \cite{nash_non_cooperative_games}, and eventually found its home as a central tool in microeconomics \cite{osbourne_rubinstein_game_theory}, as well as applications in biology and computer science.
The other branch concerned applications in logic and focussed on two-player zero-sum games, including dialogical semantics \cite{lorenzen_lorenz_dialogische_logik}, Borel games \cite{martin_borel_determinacy} and eventually game semantics in its modern sense \cite{abramsky94,hyland_ong_full_abstraction_PCF,abramsky_jagadeesan_malacaria_full_abstraction_PCF,abramsky_mccusker_game_semantics}.

Perhaps the only systematic attempt to bridge the two branches is the work of van Benthem and collaborators on game logics \cite{van_benthem_logic_games}.
Other examples of more ad-hoc bridges can be found for example in \cite{hankin_malacaria_payoffs_intensionality_abstraction,le_roux_winning_strategy_nash_equilibrium,gutierrez_wooldridge_equilibria_concurrent_games}.
The work of Pavlovic \cite{pavlovic09}, which is not specifically about game semantics, is perhaps the most closely related to this paper.

In the second half of this paper, we apply these ideas to build a compact closed category of `computable open games' by replacing the underlying dialectica category with a wave-style geometry of interaction category, specifically the Int-construction applied to the traced cartesian category of directed-complete partial orders.
(The category of directed-complete partial orders and Scott-continuous maps is a standard setting for the semantics of possibly-nonterminating recursive computations.)
Ultimately we rely on the following transport of structure result:

\begin{proposition}\label{prop:ccc-transport}
	Let $\C$ be a compact closed category, $\D$ a symmetric monoidal category and $F : \C \to \D$ a strict symmetric monoidal functor that is bijective on objects.
	Then $\D$ can be given a compact closed structure, with duals given by $F (X)^* = F (X^*)$, units by $\eta_{F (X)} = F (\eta_X)$ and counits by $\varepsilon_{F (X)} = F (\varepsilon_X)$.
\end{proposition}

\begin{proof}
	The assumption that $F$ is bijective on objects means that every object of $\D$ is uniquely assigned a dual, unit and counit.
	It is simple to check the yanking equations \cite{kelly80}:
	\begin{align*}
		&\rho_{F (X)} \circ (\id_{F (X)} \otimes \varepsilon_{F (X)}) \circ a_{F (X), F (X)^*, F (X)} \circ (\eta_{F (X)} \otimes \id_{F (X)}) \circ \lambda^{-1}_{F (X)} \\
		=\ &F (\rho_X) \circ (F (\id_X) \otimes F (\varepsilon_X)) \circ F (a_{X, X^*, X}) \circ (F (\eta_X) \otimes F (\id_X)) \circ F (\lambda^{-1}_X) \\
		=\ &F (\rho_X \circ (\id_X \otimes \varepsilon_X) \circ a_{X, X^*, X} \circ (\eta_X \otimes \id_X) \circ \lambda^{-1}_X) \\
		=\ &F (\id_X) = \id_{F (X)}
	\end{align*}
	and similarly for the other equation.
\end{proof}

The hypotheses of this theorem are already satisfied by a particular functor $\Lens \to \OG$ that identifies $\Lens$ with the subcategory of \emph{zero-player open games}.
Thus it suffices to replace the source category with one that is compact closed, while preserving the hypotheses (and the game-theoretic interpretation).

We end with a worked example, a `paradoxical' variant of matching pennies where both players have the ability and incentive to play a strategy that is contingent on the other's move - something that appears causally absurd, and can result in the play deadlocking while each player waits for the other to move first.

\section{Dialogues}

While the name `dialectica' should bring to mind dialogues in the tradition of philosophical logic (for example via Hegel's dialectics), this is apparently a coincidence.
The dialectica interpretation is named after the journal Dialectica, who published G\"odel's paper in their Paul Bernays festschrift \cite{godel58}.
But the dialectica interpretation does have a very dialectical feeling to it.

The game semantic viewpoint on G\"odel's dialectica interpretation \cite{avigad98} and de Paiva's dialectica categories \cite{depaiva_dialectica_categories_report} was described in Blass' paper that first introduced game semantics \cite{blass91}.
In this section we recall this viewpoint in detail.

We first introduce a category $\Lens$ of dialogues and strategies, which is the dialectica category over an inconsistent (1-valued) logic.

An object of $\Lens$ is a 2-stage dialogue $X^+; S^-$ in which first the System chooses $x : X$, and then the Environment chooses $s : S$, where $X$ and $S$ are any sets.
This breaks a common requirement in game semantics that the Environment moves first.
We denote the dialogue $X^+; S^-$ by $\diset X S$.


Notice that the set of $P$-strategies for $\diset X S$ is $X$, and the set of $O$-strategies is $S^X$, the set of functions $X \to S$.

We introduce a monoidal product operator given by synchronous parallel play.
Specifically, the parallel play of $\diset XS$ and $\diset Y R$ is the 4-stage dialogue $X^+; Y^+; R^-; S^-$.
This peculiar ordering of moves, with the right-hand dialogue being played in the middle of the left-hand dialogue, is characteristic of dialectica.
This 4-stage dialogue is strategically equivalent to the 2-stage dialogue $(X \times Y)^+; (R \times S)^-$, so we set $\diset X S \otimes \diset Y R = \diset{X \times Y}{R \times S}$.

Next, given objects $\diset X S$ and $\diset Y R$, we consider the same 4-stage dialogue but with the players interchanged in the former.
That is, we consider the dialogue $X^-; Y^+; R^-; S^+$.
We consider this to be $\diset Y R$ played relative to $\diset X S$, and denote it by $\diset X S \to \diset Y R$.


The set of $P$-strategies for $\diset X S \to \diset Y R$ is $Y^X \times S^{X \times R}$, or isomorphically $(Y \times S^R)^X$.
The set of $O$-strategies is $X \times R^Y$.
We denote the set of $P$-strategies for $\diset X S \to \diset Y R$ by $\Lens \left( \diset X S, \diset Y R \right)$.
As the notation suggests, these are the morphisms of $\Lens$.

Given an object $\diset X S$, there is a \emph{copycat} $P$-strategy for $\diset X S \to \diset X S = X^-; X^+; S^-; S^+$.
As an element of $X^X \times S^{X \times S}$ it is the pair consisting of the identity and the projection.
This is the identity morphism for $\diset X S$.
Following \cite{abramsky_semantics_interaction} we denote this strategy by a string diagram:
\begin{center} \begin{tikzpicture}
	\node (Xm) at (0, 0) {$X^-$}; \node (Xp) at (1, 0) {$X^+$}; \node (Sm) at (2, 0) {$S^-$}; \node (Sp) at (3, 0) {$S^+$};
	\draw [->] (Xm) to [out=90, in=90] (Xp); \draw [->] (Sm) to [out=90, in=90] (Sp);
	\draw [->] (0, -1) to (Xm); \draw [->] (Xp) to (1, -1); \draw [->] (2, -1) to (Sm); \draw [->] (Sp) to (3, -1);
\end{tikzpicture} \end{center}

A major theme of this paper is that we take this notation seriously, pushing it far beyond what was originally intended.
While it is common for papers to contain a caveat that string diagrams are `officially' informal pending a coherence theorem, in this case they are far more informal than usual: it is completely unclear what category they live in, or exactly which topological moves they are invariant under.
While there is an immediate surface similarity to grammatical reductions in pregroups \cite{preller_lambek_free_compact_2_categories,coecke_sadrzadeh_clarke_discocat}, there appears to be a much deeper connection to string diagrams in the bicategory of finite product categories, Tambara modules (profunctors compatible with the cartesian product) and natural transformations \cite{boisseau_string_diagrams_optics} (see also \cite{pastro-street-doubles-monoidal-categories}), something we leave for later work.


Now suppose we are given $P$-strategies $\lambda$ for $\diset X S \to \diset Y R = X^-; Y^+; R^-; S^+$ and $\mu$ for $\diset Y R \to \diset Z Q = Y^-; Z^+; Q^-; R^+$.
There is a way to combine them to produce a $P$-strategy $\mu \circ \lambda$ for $\diset X S \to \diset Z Q = X^-; Z^+; Q^-; S^+$.
Namely, $P$ \emph{simulates} playing the two together with $O$ playing a copycat strategy for the middle moves.
That is, she simulates the 8-stage dialogue
\[ X^-; Y^+; Y^-; Z^+; Q^-; R^+; R^-; S^+ \]
with the assumption that $O$ uses a copycat strategy for the moves $Y^-$ and $R^-$.
By then hiding the $Y$ and $R$ moves we get a $P$-strategy for the required 4-stage dialogue.

We denote this $P$-strategy by the following string diagram:
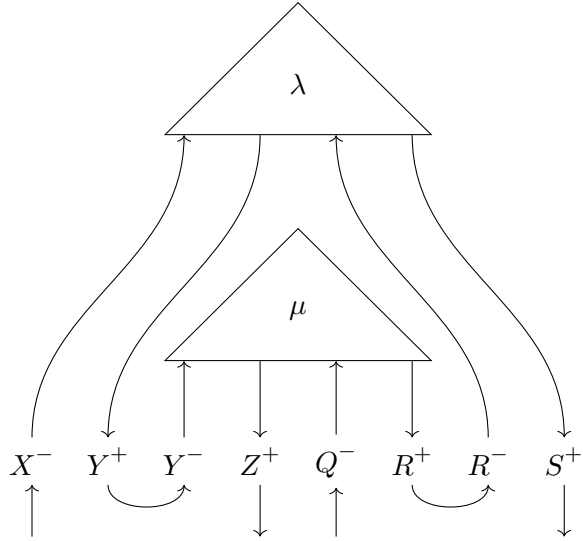
\begin{center} \begin{tikzpicture}
	\node (lambda) [isosceles triangle, shape border rotate=90, isosceles triangle apex angle=90, minimum width=3.5cm, draw] at (3.5, 5) {$\lambda$};
	\node (mu) [isosceles triangle, shape border rotate=90, isosceles triangle apex angle=90, minimum width=3.5cm, draw] at (3.5, 2) {$\mu$};
	\node (Xm) at (0, 0) {$X^-$}; \node (Yp) at (1, 0) {$Y^+$}; \node (Ym) at (2, 0) {$Y^-$}; \node (Zp) at (3, 0) {$Z^+$};
	\node (Qm) at (4, 0) {$Q^-$}; \node (Rp) at (5, 0) {$R^+$}; \node (Rm) at (6, 0) {$R^-$}; \node (Sp) at (7, 0) {$S^+$};
	\draw [->] (Xm) to [out=90, in=-90] (lambda.south -| Ym); \draw [->] (lambda.south -| Zp) to [out=-90, in=90] (Yp);
	\draw [->] (Ym) to [out=90, in=-90] (mu.south -| Ym); \draw [->] (mu.south -| Zp) to [out=-90, in=90] (Zp);
	\draw [->] (Qm) to [out=90, in=-90] (mu.south -| Qm); \draw [->] (mu.south -| Rp) to [out=-90, in=90] (Rp);
	\draw [->] (Rm) to [out=90, in=-90] (lambda.south -| Qm); \draw [->] (lambda.south -| Rp) to [out=-90, in=90] (Sp);
	\draw [->] (Yp) to [out=-90, in=-90] (Ym); \draw [->] (Rp) to [out=-90, in=-90] (Rm);
	\draw [->] (0, -1) to (Xm); \draw [->] (Zp) to (3, -1); \draw [->] (4, -1) to (Qm); \draw [->] (Sp) to (7, -1);
\end{tikzpicture} \end{center}
Whereas the cap denotes a copycat $P$-strategy, the cup denotes a copycat $O$-strategy.


A little calculation shows that if $\lambda$ is given by $v_\lambda : X \to Y$ and $u_\lambda : X \times R \to S$, and $\mu$ is given by $v_\mu : Y \to Z$ and $u_\mu : Y \times Q \to R$, then the composite is given by
\[ v_{\mu \circ \lambda} (x) = v_\mu (v_\lambda (x)) \]
and
\[ u_{\mu \circ \lambda} (x, q) = u_\lambda (x, u_\mu (v_\lambda (x), q)) \]
It is routine to check that this is associative, with identities given by copycat.
Thus $\Lens$ is indeed a category.
These equations are commonly known in functional programming as composition of lenses \cite{foster_etal_combinators_bidirectional_tree_transformations,gibbons_stevens_bidirectional_transformations}.

Given this category structure we can also make $\otimes$ into a genuine symmetric monoidal product.
Given $P$-strategies $\lambda : \diset{X_1}{S_1} \to \diset {Y_1}{R_1}$ and $\mu : \diset{X_2}{S_2} \to \diset{Y_2}{R_2}$, we can combine them to produce a $P$-strategy $\lambda \otimes \mu : \diset{X_1 \times X_2}{S_2 \times S_1} \to \diset{Y_1 \times Y_2}{R_2 \times R_1}$.

Finally, we notice that all of the above can be generalised to any base category $\C$ with finite products, replacing sets and functions, yielding a category $\Lens (\C)$ whose morphisms are strategies internal to $\C$.
Specifically, we set
\[ \Lens (\C) \left( \diset X S, \diset Y R \right) = \C (X, Y) \times \C (X \times R, S) \]
(By writing it this way, we do not need to assume that $\C$ is cartesian closed.)
The category we have been considering so far is $\Lens = \Lens (\Set)$.

\begin{proposition}
	For any category $\C$ with finite products, $\Lens (\C)$ is a symmetric monoidal category.
\end{proposition}

There is a much less obvious generalisation of $\Lens (\C)$ when $\C$ is only a monoidal category \cite{riley_categories_of_optics}, but we will not need it in this paper.

\section{Negation and $O$-strategies}

To talk about open games, we need to talk explicitly about $O$-strategies in a dialogue.
However, the categorical structure of $\Lens$ is built on $P$-strategies.
In turns out, however, that we can use $P$-strategies to talk about $O$-strategies, in a way that respects composition.

The monoidal unit of $\Lens$ is the trivial game $I = \diset 1 1 = 1^+; 1^-$.
The dialogue $I \to \diset X S$ is $1^-; X^+; S^-; 1^+$, which is strategically equivalent to $\diset X S$.
Thus the set of $P$-strategies for $I \to \diset X S$ is $X$.

If we fix a $P$-strategy $h : X$ for $\diset X S$ and another $P$-strategy $\lambda : \diset X S \to \diset Y R$, we can compose them to yield a $P$-strategy $\lambda \circ h$ for $\diset Y R$, by
\[ I \overset{h}\longrightarrow \diset X S \overset{\lambda}\longrightarrow \diset Y R \]

Succinctly, there is a functor $\mathbb V : \Lens \to \Set$ taking every object to its set of $P$-strategies, namely the covariant functor represented by $I$.
Explicitly, $\mathbb V \diset X S = X$ and $\mathbb V (\lambda) = v_\lambda$.

On the other hand, the dialogue $\diset X S \to I$ is $X^-; 1^+; 1^-; S^+$, which is equivalent to $X^-; S^+$.
This is not an object, but is $\diset X S$ with players interchanged.
Thus the set of $P$-strategies for $\diset X S \to I$ is equal to the set of $O$-strategies for $\diset X S$, namely $S^X$.


Given an $O$-strategy $k$ for $\diset Y R$ and a $P$-strategy $\lambda : \diset X S \to \diset Y R$, we obtain an $O$-strategy $k \circ \lambda$ for $\diset X S$ by
\[ \diset X S \overset{\lambda}\longrightarrow \diset Y R \overset{k}\longrightarrow I \]
In this, $O$ `hijacks' $P$'s strategy to produce an element of $S$, since $\lambda$ is a $P$-strategy for a dialogue in which $P$ plays the role of $O$ in $\diset X S$.

Succinctly, there is a functor $\mathbb K : \Lens^\op \to \Set$ taking every object to its set of $O$-strategies, namely the contravariant functor represented by $I$.
In the terminology of categorical quantum mechanics, $P$-strategies are \emph{states} and $O$-strategies are \emph{effects}.

Since an $O$-strategy for $\diset X S \to \diset Y R$ is precisely an element of $X \times Y^R$, it can be equivalently seen as a $P$-strategy for $\diset X S$ and an $O$-strategy for $\diset Y R$.
This defines a functor $\overline{\Lens} : \Lens \times \Lens^\op \to \Set$, namely
\[ \Lens \times \Lens^\op \xrightarrow{\mathbb V \times \mathbb K} \Set \times \Set \overset\times\longrightarrow \Set \]
On objects, it is concretely given by $\overline{\Lens} \left( \diset X S, \diset Y R \right) = X \times R^Y$, or more generally over a category $\C$ with finite products, $\overline{\Lens (\C)} \left( \diset X S, \diset Y R \right) = \C (1, X) \times \C (Y, R)$.

Given an $O$-strategy $\kappa = (h, k)$ for $\diset{X_1}{S_1} \to \diset{Y_1}{R_1}$, a $P$-strategy $\lambda : \diset{X_1}{S_1} \to \diset{X_2}{S_2}$ and a $P$-strategy $\mu : \diset{Y_2}{R_2} \to \diset{Y_1}{R_1}$, we obtain an $O$-strategy $\overline{\Lens} (\lambda, \mu) (h, k) = (\lambda \circ h, k \circ \mu)$ for $\diset{X_2}{S_2} \to \diset{Y_2}{R_2}$.
This is the $O$-strategy for the dialogue
\[ X_1^-; X_2^+; X_2^-; Y_2^+; Y_2^-; Y_1^+; R_1^-; R_2^+; R_2^-; S_2^+; S_2^-; S_1^+ \]
with appropriately hidden copycat moves, as given by the string diagram
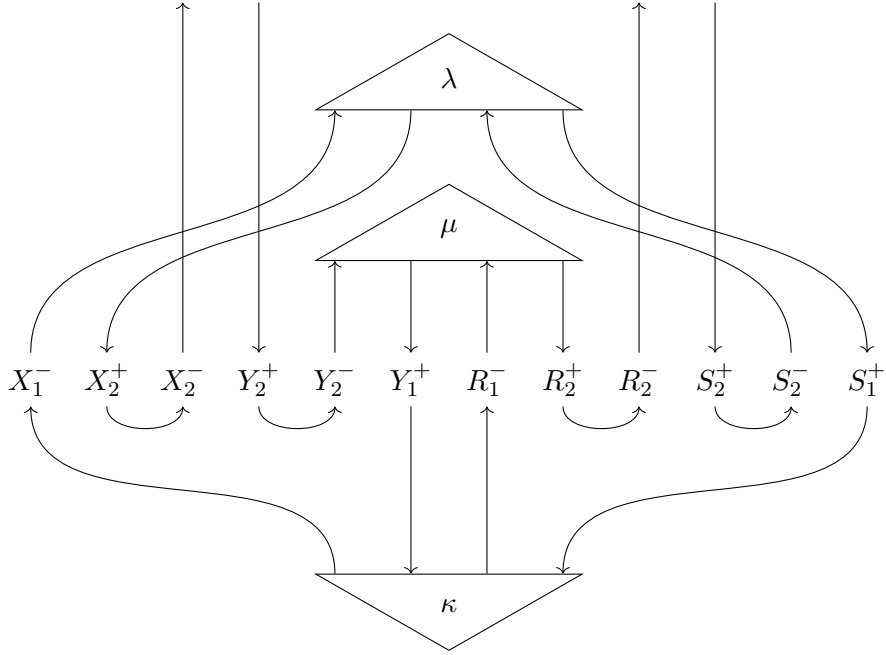
\begin{center} \begin{tikzpicture}
	\node (lambda) [isosceles triangle, shape border rotate=90, isosceles triangle apex angle=120, minimum width=3.5cm, draw] at (5.5, 4) {$\lambda$};
	\node (mu) [isosceles triangle, shape border rotate=90, isosceles triangle apex angle=120, minimum width=3.5cm, draw] at (5.5, 2) {$\mu$};
	\node (X1m) at (0, 0) {$X_1^-$}; \node (X2p) at (1, 0) {$X_2^+$}; \node (X2m) at (2, 0) {$X_2^-$}; \node (Y2p) at (3, 0) {$Y_2^+$};
	\node (Y2m) at (4, 0) {$Y_2^-$}; \node (Y1p) at (5, 0) {$Y_1^+$}; \node (R1m) at (6, 0) {$R_1^-$}; \node (R2p) at (7, 0) {$R_2^+$};
	\node (R2m) at (8, 0) {$R_2^-$}; \node (S2p) at (9, 0) {$S_2^+$}; \node (S2m) at (10, 0) {$S_2^-$}; \node (S1p) at (11, 0) {$S_1^+$};
	\node (kappa)  [isosceles triangle, shape border rotate=270, isosceles triangle apex angle=120, minimum width=3.5cm, draw] at (5.5, -3) {$\kappa$};
	\draw [->] (kappa.north -| Y2m) to [out=90, in=-90] (X1m); \draw [->] (X1m) to [out=90, in=-90] (lambda.south -| Y2m);
	\draw [->] (lambda.south -| Y1p) to [out=-90, in=90] (X2p); \draw [->] (X2p) to [out=-90, in=-90] (X2m);
	\draw [->] (X2m) to (2, 5); \draw [->] (3, 5) to (Y2p);
	\draw [->] (Y2p) to [out=-90, in=-90] (Y2m); \draw [->] (Y2m) to [out=90, in=-90] (mu.south -| Y2m);
	\draw [->] (mu.south -| Y1p) to [out=-90, in=90] (Y1p); \draw [->] (Y1p) to (kappa.north -| Y1p);
	\draw [->] (kappa.north -| R1m) to (R1m); \draw [->] (R1m) to [out=90, in=-90] (mu.south -| R1m);
	\draw [->] (mu.south -| R2p) to [out=-90, in=90] (R2p); \draw [->] (R2p) to [out=-90, in=-90] (R2m);
	\draw [->] (R2m) to (8, 5); \draw [->] (9, 5) to (S2p);
	\draw [->] (S2p) to [out=-90, in=-90] (S2m); \draw [->] (S2m) to [out=90, in=-90] (lambda.south -| R1m);
	\draw [->] (lambda.south -| R2p) to [out=-90, in=90] (S1p); \draw [->] (S1p) to [out=-90, in=90] (kappa.north -| R2p);
\end{tikzpicture} \end{center}

\section{Open games}

We can now give an equivalent definition of open games \cite{hedges_towards_compositional_game_theory,hedges_etal_compositional_game_theory} in terms of dialogues.
The treatment in this section and the next will be conceptual, with examples deferred until the end of section \ref{sec:picturing} after building up some theory.

An open game $\diset X S \to \diset Y R$ is in one dimension a dialogue played between a System and an Environment, and in another dimension it is a non-cooperative game in the sense of economics, in which several players jointly control the System while independently optimising payoffs.

An open game $\G : \diset X S \to \diset Y R$ is defined by three pieces of data:
\begin{itemize}
	\item A set $\Sigma_\G$ of strategy profiles
	\item A labelling function $\G_- : \Sigma_\G \to \Lens \left( \diset X S, \diset Y R \right)$, by which every element $\sigma : \Sigma_\G$ labels a $P$-strategy $\G_\sigma$ for the 4-stage dialogue $\diset X S \to \diset Y R$
	\item A winning condition, which is a relation between $\Sigma_\G$ and the set of $O$-strategies of $\diset X S \to \diset Y R$, namely $|\G| \subseteq \Sigma_\G \times \overline{\Lens} \left( \diset X S, \diset Y R \right)$.
\end{itemize}

We write $|\G|^\sigma_\kappa$ for $(\sigma, \kappa) \in |\G|$.
We say that $\sigma$ is a \emph{winning strategy profile} if $|\G|^\sigma_\kappa$ for all $O$-strategies $\kappa : \overline{\Lens} \left( \diset X S, \diset Y R \right)$.

We interpret $|\G|$ as an equilibrium condition.
That is, from the dialogue perspective the goal of the System is to reach equilibrium and the goal of the Environment is to prevent equilibrium.
In real examples there is rarely a winning strategy profile, and so we focus on $|\G|$ as a binary relation, or ask about winning strategy profiles for the System against a fixed $O$-strategy.

From the dialogue perspective, the order of play in an open game $\diset X S \to \diset Y R$ is:
\begin{enumerate}
	\item The Environment chooses an initial state of the game from $X$
	\item The System chooses the final state of the game from $Y$
	\item The Environment chooses payoffs for the System from $R$
	\item The System chooses payoffs for the Environment from $S$
\end{enumerate}

An $O$-strategy is a pair $\kappa = (h, k)$ where $h : X$ and $k : Y \to R$.
The \emph{history} $h$ determines the \emph{initial state} of the game.
The \emph{continuation} $k$ determines the payoffs for System given the final state.
The pair $(h, k)$ completely determines the strategic context in which the players that make up System make their choices, reducing the open game to an ordinary normal-form game.
For this reason, we also call an $O$-strategy a \emph{context} for the open game.

We only need two families of examples of open games to generate a large family of examples, corresponding roughly to extensive-form games, using the sequential and parallel play operators we will define in the next section.
These two generating families are the \emph{zero-player open games} and the \emph{decisions}, which could loosely be called \emph{one-player open games}.

The zero-player open games $\diset X S \to \diset Y R$ are in bijection with the $P$-strategies $\lambda : \diset X S \to \diset Y R$, and correspond to the situation in which the System has no strategic choices but always follows the strategy $\lambda$ like an automaton.
Specifically, the zero-player open game $\lambda$ is defined by:
\begin{itemize}
	\item The set of strategy profiles is the singleton $\Sigma_{\lambda} = \{ * \}$, where $*$ is a token representing the $P$-strategy $\lambda$
	\item The labelling function is $\lambda_{*} = \lambda$
	\item $*$ is a winning strategy profile, that is, $|\lambda|^*_\kappa$ for all $O$-strategies $\kappa$
\end{itemize}

Perhaps the only surprising part of this definition is that $*$ is a winning strategy profile.
The reason for this ultimately comes down to agreeing with Nash equilibrium on real examples.
Nash equilibrium is a \emph{negative} definition: a strategy profile should fail to be a Nash equilibrium if some particular player has positive incentive to deviate from it.
Since there are no players in $\lambda$, $*$ is declared a Nash equilibrium by default.

The second family of examples are the decisions.
There is one such open game $\D = \D_{Y | X} : \diset X 1 \to \diset Y \R$ for every nonempty set $X$ and $Y$, representing a single agent's choice from $Y$ given an observation from $X$.
In this game:
\begin{enumerate}
	\item The Environment chooses an initial state from $X$
	\item The (now unique) Player chooses a final state from $Y$
	\item The Environment chooses a payoff from $\R$
\end{enumerate}
The winning condition of this game is \emph{intensional} by being a property of the \emph{strategies} of both Player and Environment, and cannot be written in terms of the play alone.
This is because optimality in game theory is a counterfactual: \emph{if} the System had made a different choice then the resulting payoff \emph{would have} been lower.

Observe that a $P$-strategy for this game is a function $\sigma : X \to Y$, and we choose the set of strategy profiles $\Sigma_{\D_{Y | X}}$ to be precisely the set of $P$-strategies.
An $O$-strategy is a pair $(h, k)$ where $h : X$ and $k : Y \to \R$.
By definition, the Player wins this game iff $\sigma (h) \in \arg\max (k)$, that is to say, if $k (\sigma (h)) \geq k (y)$ for all $y : Y$.

This is a small shift in perspective that is quite natural from the perspective of game semantics.
In game theory there is no concept of \emph{winning}, only optimality and equilibrium.
Declaring a player to have \emph{won} if they make an optimal choice may not be meaningful as game theory, but it is appropriate terminology when combining game theory with game semantics.

Writing this out:
\begin{itemize}
	\item The set of strategy profiles is $\Sigma_\D = Y^X$
	\item The labelling function takes $\sigma : X \to Y$ to itself considered as a $P$-strategy $\D_\sigma : \diset X 1 \to \diset Y \R$, via the bijection $\Lens \left( \diset X 1, \diset Y \R \right) \cong Y^X$
	\item The winning condition is $|\D|^\sigma_{h, k}$ iff $\sigma (h) \in \arg\max (k)$
\end{itemize}

\section{Composing open games}

We can make open games into the morphisms of a symmetric monoidal category.
The two composition operators, categorical composition and tensor product, correspond to \emph{sequential play} and \emph{simultaneous play}.

Suppose we are given open games $\G : \diset X S \to \diset Y R$ and $\H : \diset Y R \to \diset Z Q$.
The sequential composition $\H \circ \G : \diset X S \to \diset Z Q$ has set of strategy profiles $\Sigma_{\H \circ \G} = \Sigma_\G \times \Sigma_\H$.
Informally, the idea is that $\G$ and $\H$ are each associated with sets $G, H$ of decisions.
Each decision $g \in G, h \in H$ has an associated set $\Sigma_g, \Sigma_h$ of strategies, and the set of strategy profiles in each case should be thought of as the set of tuples of strategies, one for each decision: $\Sigma_\G = \prod_{g \in G} \Sigma_g$ and $\Sigma_\H = \prod_{h \in H} \Sigma_h$.
The set of decisions made in a composite game is the disjoint union of the decisions made in the components, and so $\Sigma_{\H \circ \G} = \prod_{g \in G + H} \Sigma_g = \prod_{g \in G} \Sigma_g \times \prod_{h \in H} \Sigma_h = \Sigma_\G \times \Sigma_\H$.

The labelling function for a sequential composition can be defined using the underlying composition in $\Lens$: $(\H \circ \G)_{\sigma, \tau} = \H_\tau \circ \G_\sigma$.

In order to define the winning condition of $\H \circ \G$, we must modify a context for $\H \circ \G$ into contexts for $\G$ and $\H$.
We can do this using the fact that $\overline{\Lens}$ is a functor, together with the fact that we have strategy profiles for $\G$ and $\H$ available.
A strategy profile $(\sigma, \tau)$ for $\H \circ \G$ is winning (that is to say, a Nash equilibrium) against the $O$-strategy $\kappa$ iff $\sigma$ is winning in $\G$ against the $O$-strategy $\overline{\Lens} (\id, \H_\tau) (\kappa)$, and $\tau$ is winning in $\H$ against the $O$-strategy $\overline{\Lens} (\G_\sigma, \id) (\kappa)$.
That is to say,
\[ |\H \circ \G|^{\sigma, \tau}_\kappa \iff |\G|^\sigma_{\overline{\Lens} (\id, \H_\tau) (\kappa)} \wedge |\H|^\tau_{\overline{\Lens} (\G_\sigma, \id) (\kappa)} \]

This makes open games into the morphisms of a category (or, more properly, the 1-cells of a bicategory).

Next we consider simultaneous play.
Given open games $\G : \diset{X_1}{S_1} \to \diset{Y_1}{R_1}$ and $\H : \diset{X_2}{S_2} \to \diset{Y_2}{R_2}$, we combine them to form an open game
\[ \G \otimes \H : \diset{X_1 \times X_2}{S_2 \times S_1} \to \diset{Y_1 \times Y_2}{R_2 \times R_1} \]

As before the strategy profiles of $\G \otimes \H$ are pairs, $\Sigma_{\G \otimes \H} = \Sigma_\G \times \Sigma_\H$, for the same reason as before: we take the disjoint union of the set of decisions.
The strategy profile $(\sigma, \tau)$ labels the synchronous parallel play of $\G_\sigma$ and $\H_\tau$, that is, $(\G \otimes \H)_{\sigma, \tau} = \G_ \sigma \otimes \H_\tau$.

In order to define the winning condition for $\G \otimes \H$ we need to do some more work.

Given strategy profiles $\sigma : \Sigma_\G$ and $\tau : \Sigma_\H$, and an $O$-strategy $\kappa$ for $\diset{X_1 \times X_2}{S_2 \times S_1} \to \diset{Y_1 \times Y_2}{R_2 \times R_1}$, we need to `project' $\kappa$ to $\G$ and $\H$'s view of it, as $O$-strategies for $\diset{X_1}{S_1} \to \diset{Y_1}{R_1}$ and $\diset{X_2}{S_2} \to \diset{Y_2}{R_2}$.

We can indeed do this.
To produce an $O$-strategy for $\diset{X_1}{S_1} \to \diset{Y_1}{R_1}$, consider the dialogue
\[ X_1^-; Y_1^+; X_2^-; Y_2^+; R_2^-; S_2^+; R_1^-; S_1^+; \]
with the strategy
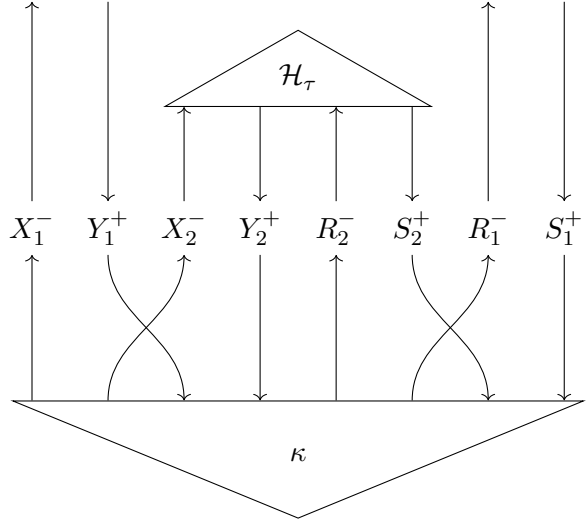
\begin{center} \begin{tikzpicture}
	\node (lambda)  [isosceles triangle, shape border rotate=90, isosceles triangle apex angle=120, minimum width=3.5cm, draw] at (3.5, 2) {$\H_\tau$};
	\node (X1m) at (0, 0) {$X_1^-$}; \node (Y1p) at (1, 0) {$Y_1^+$}; \node (X2m) at (2, 0) {$X_2^-$}; \node (Y2p) at (3, 0) {$Y_2^+$};
	\node (R2m) at (4, 0) {$R_2^-$}; \node (S2p) at (5, 0) {$S_2^+$}; \node (R1m) at (6, 0) {$R_1^-$}; \node (S1p) at (7, 0) {$S_1^+$};
	\node (kappa) [isosceles triangle, shape border rotate=270, isosceles triangle apex angle=135, minimum width=7.5cm, draw] at (3.5, -3) {$\kappa$};
	\draw [->] (kappa.north -| X1m) to (X1m); \draw [->] (X1m) to (0, 3);
	\draw [->] (1, 3) to (Y1p); \draw [->] (Y1p) to [out=-90, in=90] (kappa.north -| X2m);
	\draw [->] (kappa.north -| Y1p) to [out=90, in=-90] (X2m); \draw [->] (X2m) to (lambda.south -| X2m);
	\draw [->] (lambda.south -| Y2p) to (Y2p); \draw [->] (Y2p) to (kappa.north -| Y2p);
	\draw [->] (kappa.north -| R2m) to (R2m); \draw [->] (R2m) to (lambda.south -| R2m);
	\draw [->] (lambda.south -| S2p) to (S2p); \draw [->] (S2p) to [out=-90, in=90] (kappa.north -| R1m);
	\draw [->] (kappa.north -| S2p) to [out=90, in=-90] (R1m); \draw [->] (R1m) to (6, 3);
	\draw [->] (7, 3) to (S1p); \draw [->] (S1p) to (kappa.north -| S1p);
\end{tikzpicture} \end{center}
We call this $O$-strategy $\H_\tau / \kappa$.
When $\kappa = ((h_1, h_2), k)$ for $k : Y_1 \times Y_2 \to R_2 \times R_1$, we write $\H_\tau / \kappa = (h_1, k_1^{h_2, \H_\tau})$.
Concretely, the new continuation is $k_1^{h_2, \H_\tau} (y_1) = k (y_1, v_{\H_\tau} (h_2))_2$.

Similarly, we can produce an $O$-strategy $\G_\sigma \setminus \kappa$ for $\diset{X_2}{S_2} \to \diset{Y_2}{R_2}$ by considering the same dialogue with the strategy
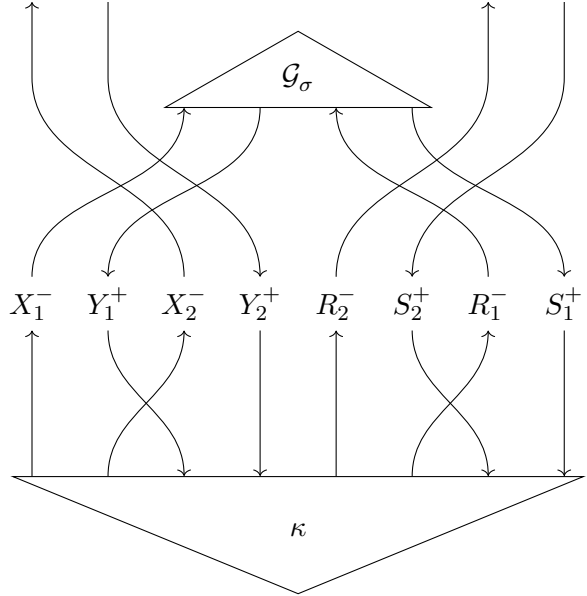
\begin{center} \begin{tikzpicture}
	\node (lambda)  [isosceles triangle, shape border rotate=90, isosceles triangle apex angle=120, minimum width=3.5cm, draw] at (3.5, 3) {$\G_\sigma$};
	\node (X1m) at (0, 0) {$X_1^-$}; \node (Y1p) at (1, 0) {$Y_1^+$}; \node (X2m) at (2, 0) {$X_2^-$}; \node (Y2p) at (3, 0) {$Y_2^+$};
	\node (R2m) at (4, 0) {$R_2^-$}; \node (S2p) at (5, 0) {$S_2^+$}; \node (R1m) at (6, 0) {$R_1^-$}; \node (S1p) at (7, 0) {$S_1^+$};
	\node (kappa) [isosceles triangle, shape border rotate=270, isosceles triangle apex angle=135, minimum width=7.5cm, draw] at (3.5, -3) {$\kappa$};
	\draw [->] (kappa.north -| X1m) to (X1m); \draw [->] (X1m) to [out=90, in=-90] (lambda.south -| X2m);
	\draw [->] (lambda.south -| Y2p) to [out=-90, in=90] (Y1p); \draw [->] (Y1p) to [out=-90, in=90] (kappa.north -| X2m);
	\draw [->] (kappa.north -| Y1p) to [out=90, in=-90] (X2m); \draw [->] (X2m) to [out=90, in=-90] (0, 3) to (0, 4);
	\draw [->] (1, 4) to (1, 3) to [out=-90, in=90] (Y2p); \draw [->] (Y2p) to (kappa.north -| Y2p);
	\draw [->] (kappa.north -| R2m) to (R2m); \draw [->] (R2m) to [out=90, in=-90] (6, 3) to (6, 4);
	\draw [->] (7, 4) to (7, 3) to [out=-90, in=90] (S2p); \draw [->] (S2p) to [out=-90, in=90] (kappa.north -| R1m);
	\draw [->] (kappa.north -| S2p) to [out=90, in=-90] (R1m); \draw [->] (R1m) to [out=90, in=-90] (lambda.south -| R2m);
	\draw [->] (lambda.south -| S2p) to [out=-90, in=90] (S1p); \draw [->] (S1p) to (kappa.north -| S1p);
\end{tikzpicture} \end{center}
When $\kappa = ((h_1, h_2), k)$ we write $\G_\sigma \setminus \kappa = (h_2, k_2^{h_1, \G_\sigma})$, where $k_2^{h_1, \G_\sigma} (y_2) = k (v_{\G_\sigma} (h_1), y_2)_1$.

With this, we can finally define the winning condition for $\G \otimes \H$: The strategy profile $(\sigma, \tau)$ is winning against $\kappa$ in $\G \otimes \H$ iff $\sigma$ is winning against $\H_\tau / \kappa$ in $\G$ and $\tau$ is winning against $\G_\sigma \setminus \kappa$ in $\H$, that is to say,
\[ |\G \otimes \H|^{\sigma, \tau}_\kappa \iff |\G|^\sigma_{\H_\tau / \kappa} \wedge |\H|^\tau_{\G_\sigma \setminus \kappa} \]

\begin{proposition}
	There is a symmetric monoidal (bi)category $\OG$ whose objects are pairs of sets and morphisms are open games.
\end{proposition}

Although $\OG$ should properly be thought of as a bicategory with 2-cells given by appropriately compatible functions between sets of strategy profiles, this is an uninteresting technicality and we will instead quotient out these 2-cells, treating open games as defined only up to compatible bijections of strategy profiles.
The details of this can be found in \cite{hedges_morphisms_open_games}.


\section{Picturing open games}\label{sec:picturing}

Since open games are the morphisms of a monoidal category, we can depict them by string diagrams, and in fact this turns out to be invaluable for working with them in practice.
As a special case of this we also obtain string diagrams for the monoidal category of $P$-strategies, which are equivalently the wide subcategory of zero-player open games.
These diagrams should not be confused with the (less well understood) diagrams for dialogues that have appeared so far in this paper, which are very different, although to some extent it is possible to translate between them.
This section contains nothing new, but is included from \cite{hedges_etal_compositional_game_theory} for completeness.

A $P$-strategy $\lambda : \diset X S \to \diset Y R$, viewed as a zero-player open game, is depicted as a string diagram
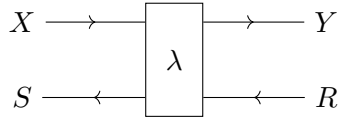
\begin{center} \begin{tikzpicture}
	\node (X) at (-2, .5) {$X$}; \node (Y) at (2, .5) {$Y$}; \node (R) at (2, -.5) {$R$}; \node (S) at (-2, -.5) {$S$};
	\node [rectangle, minimum height=1.5cm, minimum width=.75cm, draw] (G) at (0, 0) {$\lambda$};
	\draw [->-] (X) to (G.west |- X); \draw [->-] (G.east |- Y) to (Y); \draw [->-] (R) to (G.east |- R); \draw [->-] (G.west |- S) to (S);
\end{tikzpicture} \end{center}
We regard the forwards-oriented strings labelled $X$ and $Y$ as respectively representing the objects $\diset X 1$ and $\diset Y 1$, and the backwards-oriented strings labelled $R$ and $S$ are respectively representing the objects $\diset 1 R$ and $\diset 1 S$.
Thus we are implicitly using the natural isomorphisms $\diset X 1 \otimes \diset 1 S = \diset{X \times 1}{S \times 1} \cong \diset X S$ and $\diset Y 1 \otimes \diset 1 R = \diset{Y \times 1}{R \times 1} \cong \diset Y R$.

As special cases of this, a function $f : X \to Y$ can be regarded as a $P$-strategy and as a zero-player open game either covariantly as $f : \diset X 1 \to \diset Y 1$, or contravariantly as $f^* : \diset 1 Y \to \diset 1 X$.
We depict these respectively with the diagrams
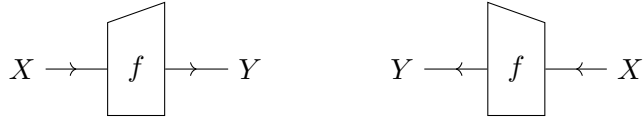
\begin{center} \begin{tikzpicture}
	\node (X1) at (0, 0) {$X$}; \node (Y1) at (3, 0) {$Y$};
	\node [trapezium, trapezium left angle=0, trapezium right angle=75, shape border rotate=90, trapezium stretches=true, minimum height=.75cm, minimum width=1.5cm, draw] (f1) at (1.5, 0) {$f$};
	\draw [->-] (X1) to (f1); \draw [->-] (f1) to (Y1);
	\node (X2) at (8, 0) {$X$}; \node (Y2) at (5, 0) {$Y$};
	\node [trapezium, trapezium left angle=75, trapezium right angle=0, shape border rotate=270, trapezium stretches=true, minimum height=.75cm, minimum width=1.5cm, draw] (f2) at (6.5, 0) {$f$};
	\draw [->-] (X2) to (f2); \draw [->-] (f2) to (Y2);
\end{tikzpicture} \end{center}
As a further special case, the liftings $\Delta_X : \diset X 1 \to \diset{X \times X}{1}$ and $\Delta_X^* : \diset{1}{X \times X} \to \diset 1 X$ of the copy functions are given the special syntax
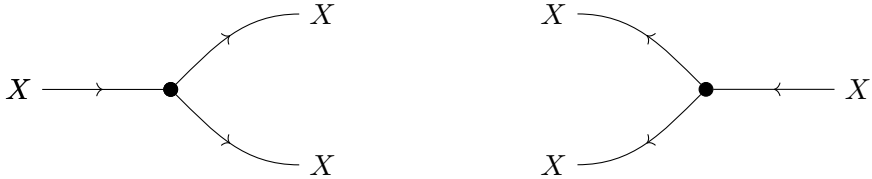
\begin{center} \begin{tikzpicture}
	\node (X2) at (6, 1) {$X$}; \node [circle, scale=0.5, fill=black, draw] (m) at (8, 1) {};
	\node (X2) at (6, 1) {$X$}; \node [circle, scale=0.5, fill=black, draw] (m) at (8, 1) {};
	\node (X3) at (10, 2) {$X$}; \node (X4) at (10, 0) {$X$};
	\draw [->-] (X2) to (m); \draw [->-] (m) to [out=45, in=180] (X3); \draw [->-] (m) to [out=-45, in=180] (X4);
\end{tikzpicture} \qquad\qquad\qquad \begin{tikzpicture}
	\node (X2) at (6, -1) {$X$}; \node (X3) at (6, -3) {$X$};
	\node [circle, scale=0.5, fill=black, draw] (m) at (8, -2) {}; \node (X4) at (10, -2) {$X$};
	\draw [->-] (X4) to (m);
	\draw [->-] (m) to [out=135, in=0] (X2); \draw [->-] (m) to [out=-135, in=0] (X3);
\end{tikzpicture} \end{center}

For any set $X$ there is a copycat $P$-strategy $\varepsilon_X : \diset X X \to I$, arising from the copycat $O$-strategy for $X^+; X^-$ via the representation $\mathbb K \cong \Lens (-, I)$.
We depict this $P$-strategy and the corresponding zero player open game by a cap
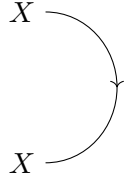
\begin{center} \begin{tikzpicture}
	\node (X1) at (0, 2) {$X$}; \node (X2) at (0, 0) {$X$};
	\draw [->-] (X1) to [out=0, in=90] (1.25, 1) to [out=-90, in=0] (X2);
\end{tikzpicture} \end{center}
However, there is no corresponding family of cups $\eta_X : I \to \diset X X$, so we do not allow wires to bend the other way in our diagrams.

The $P$-strategies $\varepsilon_X : \diset X X \to I$ are dinatural in $X$, which means that for any function $f : X \to Y$ the diagram
\[ \begin{tikzcd}
	\diset X Y \ar[rr, "f \otimes \id_{\diset 1 Y}"] \ar[dd, "\id_{\diset X 1} \otimes f^*"'] && \diset Y Y \ar[dd, "\varepsilon_Y"] \\ \\
	\diset X X \ar[rr, "\varepsilon_X"] && I
\end{tikzcd} \]
in $\Lens$ commutes.
In string diagrams, this equation is depicted
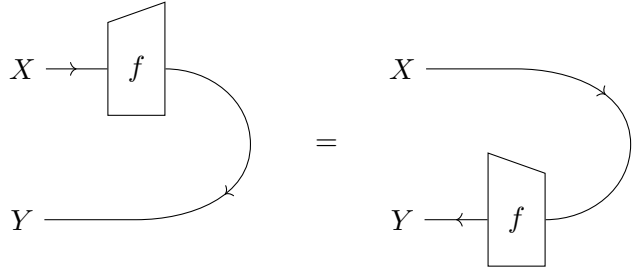
\begin{center} \begin{tikzpicture}
	\node (X1) at (0, 2) {$X$}; \node (Y1) at (0, 0) {$Y$};
	\node [trapezium, trapezium left angle=0, trapezium right angle=75, shape border rotate=90, trapezium stretches=true, minimum height=.75cm, minimum width=1.5cm, draw] (f1) at (1.5, 2) {$f$};
	\draw [->-] (X1) to (f1); \draw [->-] (f1) to [out=0, in=90] (3, 1) to [out=-90, in=0] (1.5, 0) to (Y1);
	\node at (4, 1) {$=$};
	\node (X2) at (5, 2) {$X$}; \node (Y2) at (5, 0) {$Y$};
	\node [trapezium, trapezium left angle=75, trapezium right angle=0, shape border rotate=270, trapezium stretches=true, minimum height=.75cm, minimum width=1.5cm, draw] (f2) at (6.5, 0) {$f$};
	\draw [->-] (X2) to (6.5, 2) to [out=0, in=90] (8, 1) to [out=-90, in=0] (f2); \draw [->-] (f2) to (Y2);
\end{tikzpicture} \end{center}
The reader should visualise $f$ flipping over rather than rotating within the plane.
This comes from the convention that $\otimes$ reverses the contravariant part of an object, and corresponds to the choice of algebraic rather than diagrammatic transpose in \cite[section 4.2.2]{coecke_kissinger_picturing_quantum_processes}.

This can be seen as a sort of partial duality, which is defined on all objects by $\diset X S^* = \diset S X$ (which is interchange of players in a dialogue) and on $P$-strategies of the form $f$ and $f^*$, but on no other open games besides these.
In the last section of this paper we will extend this to a fully-fledged duality in the sense of compact closure.

A decision $\D_{Y | X} : \diset X 1 \to \diset Y \R$ and its special case $\D_Y = \D_{Y | 1} : I \to \diset Y \R$ are respectively depicted
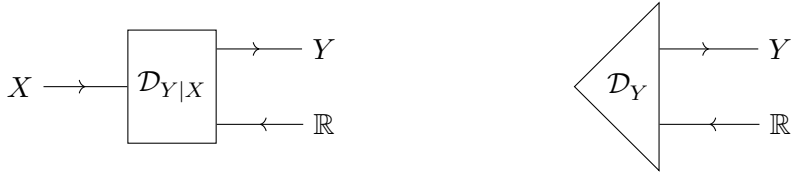
\begin{center} \begin{tikzpicture}
	\node (X) at (-2, 0) {$X$}; \node (Y) at (2, .5) {$Y$}; \node (R) at (2, -.5) {$\R$};
	\node [rectangle, minimum height=1.5cm, minimum width=.75cm, draw] (G) at (0, 0) {$\D_{Y | X}$};
	\draw [->-] (X) to (G); \draw [->-] (G.east |- Y) to (Y); \draw [->-] (R) to (G.east |- R);
	\node (Y2) at (8, .5) {$Y$}; \node (R2) at (8, -.5) {$\R$};
	\node [isosceles triangle, isosceles triangle apex angle=90, shape border rotate=180, minimum width=1.5cm, draw] (D2) at (6, 0) {$\D_Y$};
	\draw [->-] (D2.east |- Y2) to (Y2); \draw [->-] (R2) to (D2.east |- R2);
\end{tikzpicture} \end{center}

The string diagrams built from these diagram elements correspond to the open games generated from zero-player open games and decisions by sequential and parallel composition.
Given a pair of payoff matrices $U : X \times Y \to \R \times \R$, the resulting bimatrix game corresponds to the diagram
	
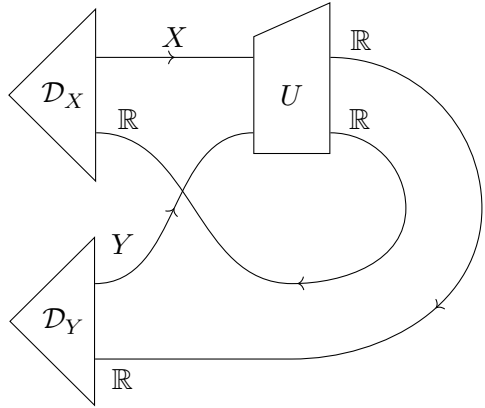
\begin{center} \begin{tikzpicture}
		\node [isosceles triangle, isosceles triangle apex angle=90, shape border rotate=180, minimum width=2cm, draw] (D1) at (0, 3) {$\D_X$};
		\node [isosceles triangle, isosceles triangle apex angle=90, shape border rotate=180, minimum width=2cm, draw] (D2) at (0, 0) {$\D_Y$};
		\node [trapezium, trapezium left angle=0, trapezium right angle=75, shape border rotate=90, trapezium stretches=true, minimum height=1cm, minimum width=2cm, draw] (U) at (3, 3) {$U$};
		\node (d1) at (0, -.5) {}; \node (d2) at (0, .5) {}; \node (d3) at (0, 2.5) {}; \node (d4) at (0, 3.5) {}; \node (d5) at (0, 1) {}; \node (d6) at (0, 2) {};
		\draw [->-] (D1.east |- d4) to node [above] {$X$} (U.west |- d4);
		\draw [->-] (D2.east |- d2) to [out=0, in=180] node [above=5pt, very near start] {$Y$} (U.west |- d3);
		\draw [->-] (U.east |- d4) to [out=0, in=90] node [above, very near start] {$\R$} (5.5, 1.5) to [out=-90, in=0] (3, -.5) to [out=180, in=0] node [below, very near end] {$\R$} (D2.east |- d1);
		\draw [->-] (U.east |- d3) to [out=0, in=90] node [above, near start] {$\R$} (4.5, 1.5) to [out=-90, in=0] (3, .5) to [out=180, in=0] node [above, very near end] {$\R$} (D1.east |- d3);
	\end{tikzpicture} \end{center}
in the sense that the scalar open game $\G : I \to I$ defined by the diagram has as strategy profiles $\Sigma_\G = X \times Y$ the pure strategy profiles of the bimatrix game, and as equilibria the pure strategy Nash equilibria of the bimatrix game: $|\G|^{x, y}$ holds iff $x \in \arg\max_{x'} U_1 (x', y)$ and $y \in \arg\max_{y'} U_2 (x, y')$.
This directly generalises to normal-form games with any finite number of players.

Similarly, the diagram
	
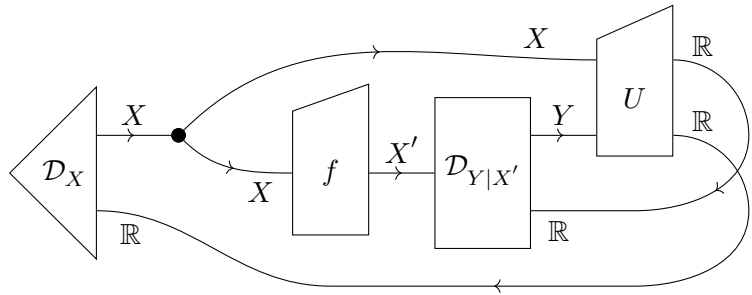
\begin{center} \begin{tikzpicture}
		\node [isosceles triangle, isosceles triangle apex angle=90, shape border rotate=180, minimum width=2cm, draw] (D1) at (.5, 0) {$\D_X$};
		\node [circle, scale=.5, fill=black, draw] (m) at (2, .5) {};
		\node [rectangle, minimum height=2cm, draw] (D2) at (6, 0) {$\D_{Y | X'}$};
		\node [trapezium, trapezium left angle=0, trapezium right angle=75, shape border rotate=90, trapezium stretches=true, minimum height=1cm, minimum width=2cm, draw] (f) at (4, 0) {$f$};
		\node [trapezium, trapezium left angle=0, trapezium right angle=75, shape border rotate=90, trapezium stretches=true, minimum height=1cm, minimum width=2cm, draw] (q) at (8, 1) {$U$};
		\node (d1) at (0, -.5) {}; \node (d2) at (0, 1.5) {};
		\draw [->-] (D1.east |- m) to node [above] {$X$} (m);
		\draw [->-] (m) to [out=45, in=180] node [above, very near end] {$X$} (q.west |- d2);
		\draw [->-] (m) to [out=-45, in=180] node [below, near end] {$X$} (f);
		\draw [->-] (f) to node [above] {$X'$} (D2); \draw [->-] (D2.east |- m) to node [above] {$Y$} (q.west |- m);
		\draw [->-] (q.east |- d2) to [out=0, in=90] node [above, near start] {$\R$} (9.5, .5) to [out=-90, in=0] (8, -.5) to [out=180, in=0] node [below, near end] {$\R$} (D2.east |- d1);
		\draw [->-] (q.east |- m) to [out=0, in=90] node [above, near start] {$\R$} (9.5, -.5) to [out=-90, in=0] (8, -1.5) to (4, -1.5) to [out=180, in=0] node [below, very near end] {$\R$} (D1.east |- d1);
	\end{tikzpicture} \end{center}
describes a 2-player sequential game in which the first player chooses $x$ and then the second player chooses $y$ after observing $f (x)$ for some function $f : X \to X'$, which is equivalently an extensive form with player 2's information sets given by the equivalence relation on $X'$ induced by $f$.
As special cases, if $f$ is the identity function then the $f$ node can be drawn as a plain wire and we obtain a game of perfect information, and if $f : X \to 1$ is the delete function then $f$ cancels with the copy function and the diagram can be deformed into the previous one to obtain a bimatrix game.
The scalar game $\G : I \to I$ defined by the diagram has $\Sigma_\G = X \times Y^{X'}$ given by the pure strategy profiles, and $|\G|^{x, f}$ holds iff $x \in \arg\max_{x'} U_1 (x', f (x'))$ and $f (x) \in \arg\max_{y'} U_2 (x, y')$.
Notice that these are the Nash equilibria of the extensive form game, rather than the subgame perfect equilibria.
Again, this generalises to extensive form games with any finite number of players.

\section{Dialogues and wave-style geometry of interaction}

In order to obtain a connection between the dialectica and $\Int$ constructions, we need to apply the $\Int$ construction to categories that are traced cartesian monoidal.
This is \emph{wave-style geometry of interaction}, so-called because every point in our string diagrams is consistently assigned a value \cite{abramsky_retracing_paths_process_algebra}.
(It is contrasted with \emph{particle-style GoI}, which applies to monoidal categories built on a coproduct and in which we imagine a token moving around the diagram.)

Game-semantic interpretations of wave-style GoI have not been widely considered.
In this section we suggest such an interpretation that will be suitable for our purposes.

The $\Int$-construction can be defined over any traced monoidal category $\C$, but we restrict to traced cartesian categories.
These are equivalent to \emph{Conway cartesian categories}, or cartesian categories with a natural family of fixpoint operators \cite{hawegawa_recursion_cyclic_sharing} (see also \cite{ponto_shulman_traces_symmetric_monoidal_categories}).
A canonical example is the category $\DCPO$ of directed-complete partial orders and Scott-continuous maps.

By definition, an object of the category $\Int (\C)$ is a pair $\diset X S$ of objects of $\C$, and a morphism $\diset X S \to \diset Y R$ in $\Int (\C)$ is a morphism $X \times R \to Y \times S$ in $\C$.
Since $\C$ is cartesian monoidal, a morphism $\diset X S \to \diset Y R$ is equivalently a pair of morphisms $X \times R \to Y$ and $X \times R \to S$.

The identity on $\diset X S$ in $\Int (\C)$ is the identity on $X \times S$ in $\C$.
The composition of $\lambda : \diset X S \to \diset Y R$ and $\mu : \diset Y R \to \diset Z Q$ in $\Int (\C)$ is given by
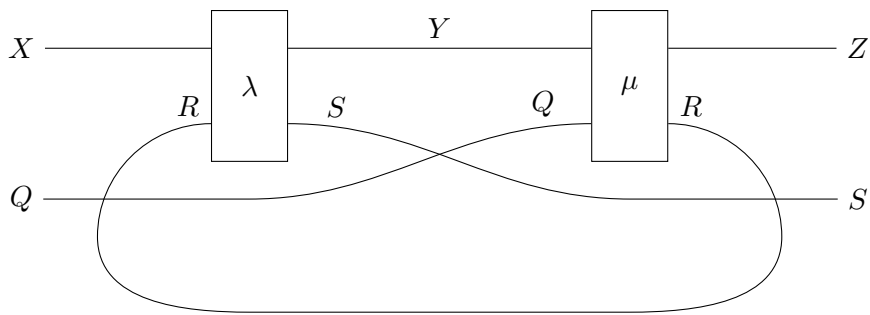
\begin{center} \begin{tikzpicture}
	\node (X) at (0, .5) {$X$}; \node (Z) at (11, .5) {$Z$};
	\node (Q) at (0, -1.5) {$Q$}; \node (S) at (11, -1.5) {$S$};
	\node (d) at (0, -.5) {};
	\node (f) [rectangle, minimum height=2cm, minimum width=1cm, draw] at (3, 0) {$\lambda$};
	\node (g) [rectangle, minimum height=2cm, minimum width=1cm, draw] at (8, 0) {$\mu$};
	\draw [-] (X) to (f.west |- X); \draw [-] (f.east |- X) to node [above] {$Y$} (g.west |- Z); \draw [-] (g.east |- Z) to (Z);
	\draw [-] (Q) to (3, -1.5) to [out=0, in=180] node [above, very near end] {$Q$} (g.west |- d);
	\draw [-] (f.east |- d) to [out=0, in=180] node [above, very near start] {$S$} (8, -1.5) to (S);
	\draw [-] (g.east |- d) to [out=0, in=90] node [above, very near start] {$R$} (10, -2) to [out=-90, in=0] (8, -3) to (3, -3) to [out=180, in=-90] (1, -2) to [out=90, in=180] node [above, very near end] {$R$} (f.west |- d);
\end{tikzpicture} \end{center}
in $\C$, using the string diagram language for traced monoidal categories \cite[section 5.7]{selinger11}.

The monoidal product of $\Int (\C)$ is defined on objects by $\diset{X_1}{S_1} \otimes \diset{X_2}{S_2} = \diset{X_1 \otimes X_2}{S_2 \otimes S_1}$, with the obvious definition on morphisms.
As is well known, $\Int (\C)$ can be equipped with the structure of a compact closed category, which satisfies the universal property of being the free compact closed category on the traced monoidal category $\C$.
Note that there are two different conventions in use: we follow \cite{joyal96}, which defines $\otimes$ with a twist in the contravariant place, rather than \cite{abramsky_retracing_paths_process_algebra} which does not.

The idea of interpreting objects and morphisms of $\Int (\C)$ as dialogues is to view them as repeated play of the corresponding dialogues for $\Lens (\C)$, starting from $\bot$ and converging to a fixpoint, after which the play terminates and all moves except the final ones are hidden.

We view the object $\diset X S$ as a dialogue
\[ X^+; S^-; X^+; S^-; \cdots \]
We do not allow arbitrary strategies, but restrict the allowed $P$-strategies to $\C$-morphisms $S \to X$, and the allowed $O$-strategies to the $\C$-morphisms $X \to S$.
Given such a pair of strategies $(h, k)$, the play that results is by definition
\[ \bot_X; \bot_S; h (\bot_S); k (\bot_X); h (k (\bot_X)); k (h (\bot_S)); \cdots \]
When $\C$ is $\DCPO$ or another suitable category, this play stabilises after finitely many stages to the $(x, s)$ that is the least fixpoint of the recursion $x = h (s)$, $s = k (x)$.
By hiding the approximating moves, we consider the play resulting from $(h, k)$ to be $(x, s)$.


Given objects $\diset X S$ and $\diset Y R$, the dialogue $\diset X S \to \diset Y R$ is
\[ X^-; Y^+; R^-; S^+; X^-; Y^+; R^-; S^+; \cdots \]
We restrict the allowed $P$-strategies to $\C$-morphisms $X \times R \to Y \times S$ and the allowed $O$-strategies to $\C$-morphisms $Y \times S \to X \times R$.
Given a $P$-strategy $\lambda = \left< v, u \right>$ and an $O$-strategy $\kappa = \left< h, k \right>$, the resulting play is
\begin{align*}
	x_0 &= \bot_X; & y_0 &= \bot_Y; & r_0 &= \bot_R; & s_0 &= u (x_0, r_0); \\
	x_{n + 1} &= h (y_n, s_n); & y_{n + 1} &= v (x_{n + 1}, r_n); & r_{n + 1} &= k (y_{n + 1}, s_n); & s_{n + 1} &= u (x_{n + 1}, r_{n + 1})
\end{align*}
This stabilises after finitely many stages to $(x, y, r, s)$ which is the least fixpoint of $(x, r) = \kappa (y, s)$, $(y, s) = \lambda (x, r)$.
Again we hide the approximating moves so that $(x, y, r, s)$ is the visible play.

\section{From dialectica to geometry of interaction}

The previous section suggests that every $P$-strategy in $\Lens (\C) \left( \diset X S, \diset Y R \right)$ can also be viewed as a $P$-strategy in $\Int (\C) \left( \diset X S, \diset Y R \right)$.
We could also discover this fact simply by inspecting the definitions, without thinking in terms of dialogues.

Incidentally, \cite{hasuo_hoshino_semantics_higher_order_quantum_computation_geometry_interaction} refers to the Int-construction as ``bidirectional computation'', a technical term that usually refers to lenses and related constructions (e.g. \cite{gibbons_stevens_bidirectional_transformations}).

In this section we use string diagrams in the underlying category $\C$.
This is the language of traced symmetric monoidal categories \cite[section 5.7]{selinger11} which are cartesian monoidal \cite[section 6.1]{selinger11}.
Implicitly, string diagrams for cartesian monoidal categories use the fact that a monoidal product is cartesian iff every object can be compatibly equipped with a commutative comonoid structure making every morphism into a comonoid homomorphism \cite{fox_coalgebras_cartesian_categories}.

\begin{proposition}
	Let $\C$ be a traced cartesian category.
	Then there is a strict monoidal functor $-^* : \Lens (\C) \to \Int (\C)$, which is identity on objects and takes the strategy $(v, u)$ to
	
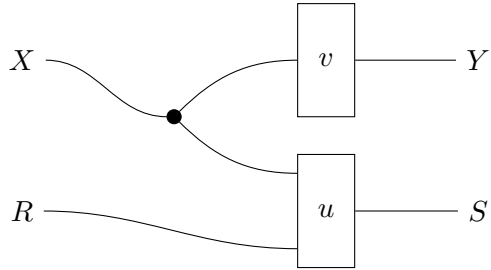
\begin{center} \begin{tikzpicture}
		\node (X) at (0, 2) {$X$}; \node (Y) at (6, 2) {$Y$}; \node (R) at (0, 0) {$R$}; \node (S) at (6, 0) {$S$};
		\node (v) [rectangle, minimum height=1.5cm, minimum width=.75cm, draw] at (4, 2) {$v$};
		\node (u) [rectangle, minimum height=1.5cm, minimum width=.75cm, draw] at (4, 0) {$u$};
		\node (copy) [circle, scale=.5, fill=black, draw] at (2, 1.25) {};
		\node (dummy1) at (0, .5) {}; \node (dummy2) at (0, -.5) {};
		\draw [-] (X) to [out=0, in=180] (copy); \draw [-] (copy) to [out=45, in=180] (v); \draw [-] (v) to (Y);
		\draw [-] (copy) to [out=-45, in=180] (u.west |- dummy1); \draw [-] (R) to [out=0, in=180] (u.west |- dummy2); \draw [-] (u) to (S);
	\end{tikzpicture} \end{center}
\end{proposition}

\begin{proof}
	The identity morphism $\diset X S \to \diset X S$ of $\Lens (\C)$ is sent to
	
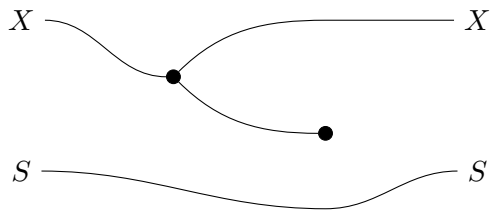
\begin{center} \begin{tikzpicture}
		\node (X) at (0, 2) {$X$}; \node (Y) at (6, 2) {$X$}; \node (R) at (0, 0) {$S$}; \node (S) at (6, 0) {$S$};
		\node (copy) [circle, scale=.5, fill=black, draw] at (2, 1.25) {};
		\node (delete) [circle, scale=.5, fill=black, draw] at (4, .5) {};
		\node (dummy1) at (0, .5) {}; \node (dummy2) at (0, -.5) {};
		\draw [-] (X) to [out=0, in=180] (copy); \draw [-] (copy) to [out=45, in=180] (4, 2) to (Y);
		\draw [-] (copy) to [out=-45, in=180] (delete); \draw [-] (R) to [out=0, in=180] (4, -.5) to [out=0, in=180] (S);
	\end{tikzpicture} \end{center}
	which is equal to the identity on $X \times S$ since the black structure is a comonoid.
	This is the identity morphism $\diset X S \to \diset XS$ of $\Int (\C)$.
	
	Next, consider morphisms $\lambda : \diset X S \to \diset Y R$ and $\mu : \diset Y R \to \diset Z Q$ in $\Lens (\C)$.
	If we compose them in $\Int (\C)$ we obtain the morphism $\mu^* \circ \lambda^*$ with string diagram
	
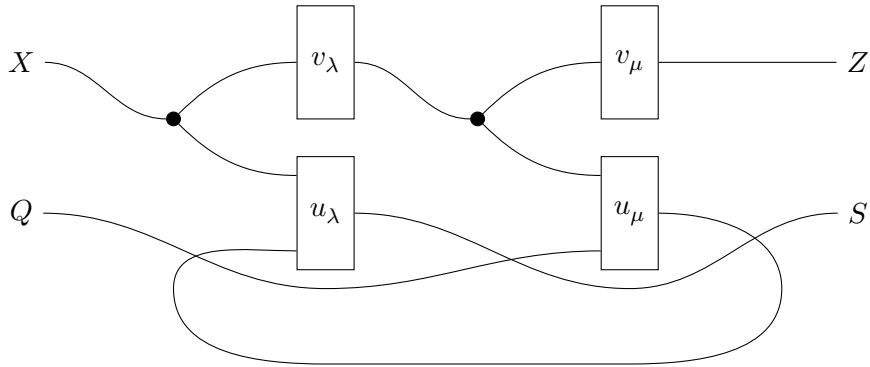
\begin{center} \begin{tikzpicture}
		\node (X) at (0, 2) {$X$}; \node (Q) at (0, 0) {$Q$}; \node (Z) at (11, 2) {$Z$}; \node (S) at (11, 0) {$S$};
		\node (copy1) [circle, scale=.5, fill=black, draw] at (2, 1.25) {};
		\node (vl) [rectangle, minimum height=1.5cm, minimum width=.75cm, draw] at (4, 2) {$v_\lambda$};
		\node (ul) [rectangle, minimum height=1.5cm, minimum width=.75cm, draw] at (4, 0) {$u_\lambda$};
		\node (copy2) [circle, scale=.5, fill=black, draw] at (6, 1.25) {};
		\node (vm) [rectangle, minimum height=1.5cm, minimum width=.75cm, draw] at (8, 2) {$v_\mu$};
		\node (um) [rectangle, minimum height=1.5cm, minimum width=.75cm, draw] at (8, 0) {$u_\mu$};
		\node (dummy1) at (0, .5) {}; \node (dummy2) at (0, -.5) {};
		\draw [-] (X) to [out=0, in=180] (copy1); \draw [-] (copy1) to [out=45, in=180] (vl); \draw [-] (vl) to [out=0, in=180] (copy2);
		\draw [-] (copy2) to [out=45, in=180] (vm); \draw [-] (vm) to (Z);
		\draw [-] (copy1) to [out=-45, in=180] (ul.west |- dummy1); \draw [-] (copy2) to [out=-45, in=180] (um.west |- dummy1);
		\draw [-] (Q) to [out=0, in=180] (4, -1) to [out=0, in=180] (um.west |- dummy2);
		\draw [-] (ul) to [out=0, in=180] (8, -1) to [out=0, in=180] (S);
		\draw [-] (um) to [out=0, in=90] (10, -1) to [out=-90, in=0] (8, -2) to (4, -2) to [out=180, in=-90] (2, -1) to [out=90, in=180] (ul.west |- dummy2);
	\end{tikzpicture} \end{center}
	Using the fact that $v_\lambda$ is a comonoid homomorphism, followed by coassociativity and symmetry of the black structure, we can transform this to
	
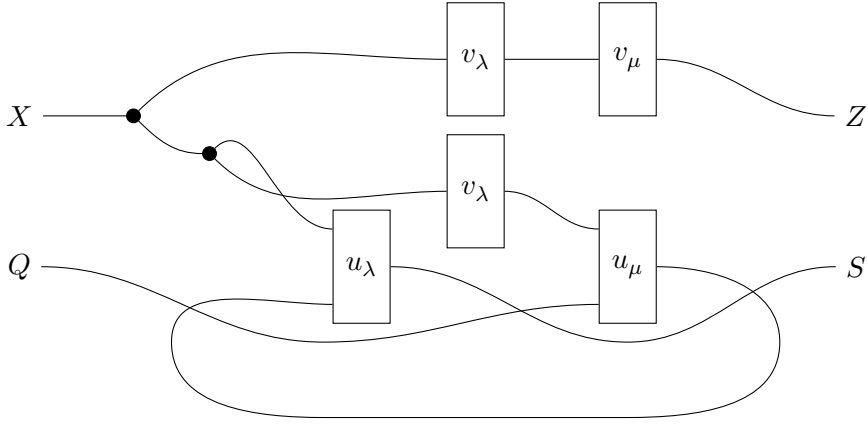
\begin{center} \begin{tikzpicture}
		\node (X) at (0, 2) {$X$}; \node (Q) at (0, 0) {$Q$}; \node (Z) at (11, 2) {$Z$}; \node (S) at (11, 0) {$S$};
		\node (copy1) [circle, scale=.5, fill=black, draw] at (1.5, 2) {};
		\node (vl1) [rectangle, minimum height=1.5cm, minimum width=.75cm, draw] at (6, 2.75) {$v_\lambda$};
		\node (vl2) [rectangle, minimum height=1.5cm, minimum width=.75cm, draw] at (6, 1) {$v_\lambda$};
		\node (ul) [rectangle, minimum height=1.5cm, minimum width=.75cm, draw] at (4.5, 0) {$u_\lambda$};
		\node (copy2) [circle, scale=.5, fill=black, draw] at (2.5, 1.5) {};
		\node (vm) [rectangle, minimum height=1.5cm, minimum width=.75cm, draw] at (8, 2.75) {$v_\mu$};
		\node (um) [rectangle, minimum height=1.5cm, minimum width=.75cm, draw] at (8, 0) {$u_\mu$};
		\node (dummy1) at (0, .5) {}; \node (dummy2) at (0, -.5) {};
		\draw [-] (X) to [out=0, in=180] (copy1); \draw [-] (copy1) to [out=45, in=180] (vl1); \draw [-] (vl1) to (vm);
		\draw [-] (vm) to [out=0, in=180] (Z);
		\draw [-] (copy1) to [out=-45, in=180] (copy2); \draw [-] (copy2) to [out=-45, in=180] (vl2);
		\draw [-] (copy2) to [out=45, in=180] (ul.west |- dummy1); \draw [-] (vl2) to [out=0, in=180] (um.west |- dummy1);
		\draw [-] (Q) to [out=0, in=180] (4, -1) to [out=0, in=180] (um.west |- dummy2);
		\draw [-] (ul) to [out=0, in=180] (8, -1) to [out=0, in=180] (S);
		\draw [-] (um) to [out=0, in=90] (10, -1) to [out=-90, in=0] (8, -2) to (4, -2) to [out=180, in=-90] (2, -1) to [out=90, in=180] (ul.west |- dummy2);
	\end{tikzpicture} \end{center}
	On the other hand, if we compose in $\Lens (\C)$, we obtain $(\mu \circ \lambda)^*$ with string diagram
	
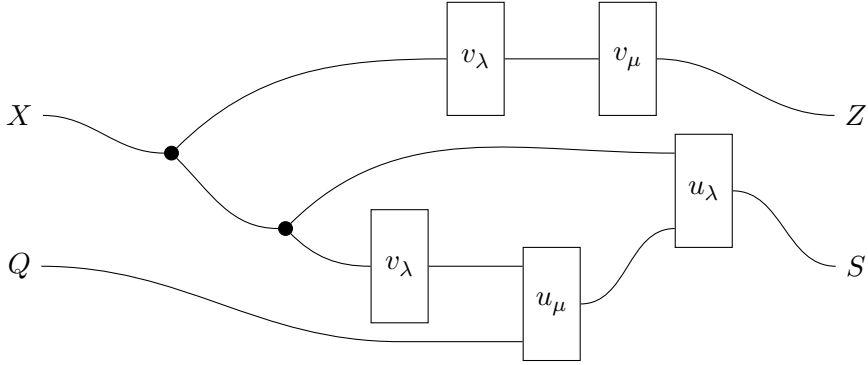
\begin{center} \begin{tikzpicture}
		\node (X) at (0, 2) {$X$}; \node (Q) at (0, 0) {$Q$}; \node (Z) at (11, 2) {$Z$}; \node (S) at (11, 0) {$S$};
		\node (copy1) [circle, scale=.5, fill=black, draw] at (2, 1.5) {};
		\node (vl1) [rectangle, minimum height=1.5cm, minimum width=.75cm, draw] at (6, 2.75) {$v_\lambda$};
		\node (vm) [rectangle, minimum height=1.5cm, minimum width=.75cm, draw] at (8, 2.75) {$v_\mu$};
		\node (copy2) [circle, scale=.5, fill=black, draw] at (3.5, .5) {};
		\node (vl2) [rectangle, minimum height=1.5cm, minimum width=.75cm, draw] at (5, 0) {$v_\lambda$};
		\node (um) [rectangle, minimum height=1.5cm, minimum width=.75cm, draw] at (7, -.5) {$u_\mu$};
		\node (ul) [rectangle, minimum height=1.5cm, minimum width=.75cm, draw] at (9, 1) {$u_\lambda$};
		\node (dummy) at (0, -1) {};
		\draw [-] (X) to [out=0, in=180] (copy1);
		\draw [-] (copy1) to [out=45, in=180] (vl1); \draw [-] (vl1) to (vm); \draw [-] (vm) to [out=0, in=180] (Z);
		\draw [-] (copy1) to [out=-45, in=180] (copy2);
		\draw [-] (copy2) to [out=45, in=180] (ul.west |- copy1);
		\draw [-] (copy2) to [out=-45, in=180] (vl2); \draw [-] (vl2) to (um.west |- vl2);
		\draw [-] (Q) to [out=0, in=180] (5, -1) to (um.west |- dummy);
		\draw [-] (um) to [out=0, in=180] (ul.west |- copy2);
		\draw [-] (ul) to [out=0, in=180] (S);
	\end{tikzpicture} \end{center}
	By inspection, we see that these string diagrams are equivalent.
	Equality of the two morphisms then follows from the coherence theorem for traced symmetric monoidal categories \cite[theorem 5.22]{selinger11}.
	
	Finally, it can be seen by inspection that the functor is strict monoidal, since $\Lens (\C)$ and $\Int (\C)$ have the same objects and the monoidal product is defined in the same way.
\end{proof}

The previous result is still true when $\C$ is an arbitrary traced monoidal category, where $\Lens (\C)$ is replaced with the more general category of optics \cite{riley_categories_of_optics}.
This was proved by Elena Di Lavore and Mario Rom\'an (private communication).


We also note that the functor $-^*$ takes the $P$-strategy $\varepsilon_X : \diset X X \to I$ to the morphism $\varepsilon_X : \diset X X \to I$ that is the counit of the compact closed structure of $\Int (\C)$.

\section{Abstracting open games}

Inspecting the definition of open games, it appears that we can define open games replacing $\Lens$ with any symmetric monoidal category $\C$ with a chosen functor $\overline\C : \C \times \C^\op \to \Set$.
This does indeed give us a \emph{category} of open games, but defining a monoidal product of open games requires an additional piece of structure, namely the ability to project individual $P$-strategies out of a $O$-strategy for a composite.
This is axiomatised by the following definition.

\begin{definition}\label{def:context}
	A \emph{context} for a symmetric monoidal category $\C$ is a symmetric monoidal functor $\overline\C : \C \times \C^\op \to \Set$ together with a natural family of functions
	\[ / : \hom_\C (X_2, Y_2) \to \left( \overline\C (X_1 \otimes X_2, Y_1 \otimes Y_2) \to \overline\C (X_1, Y_1) \right) \]
	The naturality condition required is that for all morphisms $W_1 \overset{\lambda_1}\longrightarrow X_1$, $Y_2 \overset{\nu_1}\longrightarrow Z_1$ and $W_2 \overset{\lambda_2}\longrightarrow X_2 \overset{\mu_2}\longrightarrow Y_2 \overset{\nu_2}\longrightarrow Z_2$, the diagram
	\begin{center} \begin{tikzpicture}
		\node (A) at (0, 3) {$\overline\C (Z_1 \otimes Z_2, W_1 \otimes W_2)$}; \node (B) at (7, 3) {$\overline\C (Z_1, W_1)$};
		\node (C) at (0, 0) {$\overline\C (Y_1 \otimes Y_2, X_1 \otimes X_2)$}; \node (D) at (7, 0) {$\overline\C (Y_1, X_1)$};
		\draw [->] (A) to node [above] {$(\nu_2 \circ \mu_2 \circ \lambda_2) / -$} (B); \draw [->] (B) to node [right] {$\overline\C (\nu_1, \lambda_1)$} (D);
		\draw [->] (A) to node [right] {$\overline\C (\nu_1 \otimes \nu_2, \lambda_1 \otimes \lambda_2)$} (C); \draw [->] (C) to node [above] {$\mu_2 / -$} (D);
	\end{tikzpicture} \end{center}
	in $\Set$ commutes.
\end{definition}

Using the symmetry, we can derive from this a natural family of functions
\[ \setminus : \hom_\C (X_1, Y_1) \to \left( \overline\C (X_1 \otimes X_2, Y_1 \otimes Y_2) \to \overline\C (X_2, Y_2) \right) \]
and vice versa.

The structures we defined earlier do indeed give a context on $\Lens (\C)$, namely
\[ \overline{\Lens (\C)} \left( \diset X S, \diset Y R \right) = \hom_\C (1, X) \times \hom_\C (Y, R) \]

There are trivial examples of contexts that carry no game-theoretic information, which we will ignore.
For example, we can always take $\overline\C$ to be a constant functor.
We give a second family of nontrivial examples, which we will use later.

\begin{proposition}
	Every traced symmetric monoidal category $\C$ can be equipped with the context $\overline\C (X, Y) = \hom_\C (Y, X)$, with $\lambda / \kappa$ defined by
	
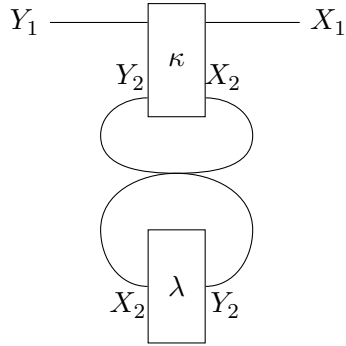
\begin{center} \begin{tikzpicture}
		\node (Y) at (-2, .5) {$Y_1$}; \node (X) at (2, .5) {$X_1$};
		\node (c) [rectangle, minimum height=1.5cm, minimum width=.75cm, draw] at (0, 0) {$\kappa$};
		\node (f) [rectangle, minimum height=1.5cm, minimum width=.75cm, draw] at (0, -3) {$\lambda$};
		\node (dummy) at (0, -.5) {};
		\draw [-] (Y) to (c.west |- Y); \draw [-] (c.east |- X) to (X);
		\draw [-] (c.east |- dummy) to [out=0, in=90] node [above, near start] {$X_2$} (1, -1) to [out=-90, in=0] (0, -1.5) to [out=180, in=90] (-1, -2.25) to [out=-90, in=180] node [below, near end] {$X_2$} (f);
		\draw [-] (f) to [out=0, in=-90] node [below, near start] {$Y_2$} (1, -2.25) to [out=90, in=0] (0, -1.5) to [out=180, in=-90] (-1, -1) to [out=90, in=180] node [above, near end] {$Y_2$} (c.west |- dummy);
	\end{tikzpicture} \end{center}
\end{proposition}

\begin{proof}
	Let $\kappa : \hom_\C (Z \otimes Z', W \otimes W')$, $\lambda_1 : W_1 \to X_1$, $\nu_1 : Y_2 \to Z_1$ and $W_2 \overset{\lambda_2}\longrightarrow X_2 \overset{\mu_2}\longrightarrow Y_2 \overset{\nu_2}\longrightarrow Z_2$.
	We chase\footnote{The author has named this proof technique `string diagram chasing', i.e. chasing an element around a commuting diagram whose nodes are all formed from homsets in a monoidal category.} the context $\kappa$ around the commuting diagram in definition \ref{def:context}.
	By the upper route we obtain
	
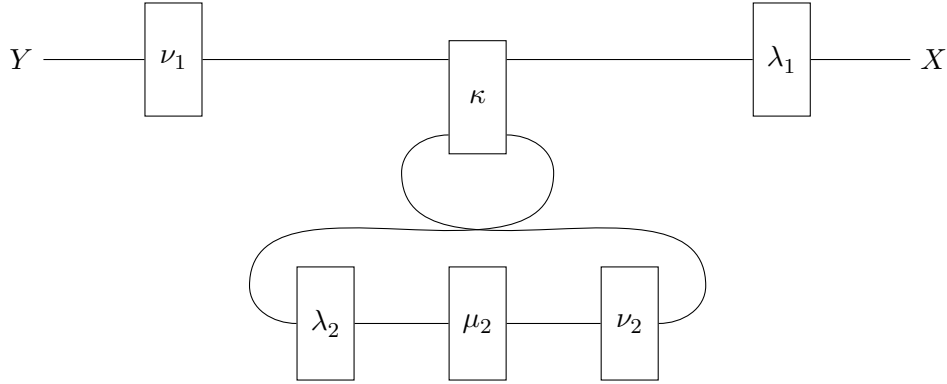
\begin{center} \begin{tikzpicture}
		\node (Y) at (-6, 3.5) {$Y$}; \node (X) at (6, 3.5) {$X$};
		\node (c) [rectangle, minimum height=1.5cm, minimum width=.75cm, draw] at (0, 3) {$\kappa$};
		\node (h) [rectangle, minimum height=1.5cm, minimum width=.75cm, draw] at (-4, 3.5) {$\nu_1$};
		\node (f) [rectangle, minimum height=1.5cm, minimum width=.75cm, draw] at (4, 3.5) {$\lambda_1$};
		\node (f') [rectangle, minimum height=1.5cm, minimum width=.75cm, draw] at (-2, 0) {$\lambda_2$};
		\node (g') [rectangle, minimum height=1.5cm, minimum width=.75cm, draw] at (0, 0) {$\mu_2$};
		\node (h') [rectangle, minimum height=1.5cm, minimum width=.75cm, draw] at (2, 0) {$\nu_2$};
		\node (dummy) at (0, 2.5) {};
		\draw [-] (Y) to (h); \draw [-] (h) to (c.west |- h); \draw [-] (c.east |- f) to (f);  \draw [-] (f) to (X);
		\draw [-] (c.east |- dummy) to [out=0, in=90] (1, 2) to [out=-90, in=90] (-3, .5) to [out=-90, in=180] (f');
		\draw [-] (c.west |- dummy) to [out=180, in=90] (-1, 2) to [out=-90, in=90] (3, .5) to [out=-90, in=0] (h');
		\draw [-] (f') to (g'); \draw [-] (g') to (h');
	\end{tikzpicture} \end{center}
	and by the lower route we obtain
	
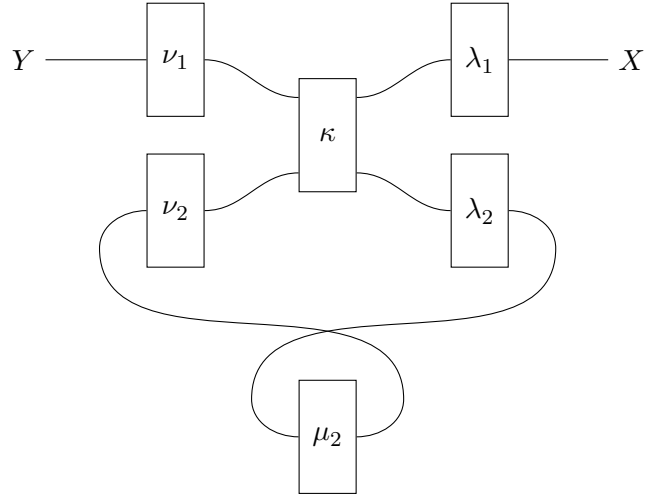
\begin{center} \begin{tikzpicture}
		\node (Y) at (-4, 4) {$Y$}; \node (X) at (4, 4) {$X$};
		\node (c) [rectangle, minimum height=1.5cm, minimum width=.75cm, draw] at (0, 3) {$\kappa$};
		\node (h) [rectangle, minimum height=1.5cm, minimum width=.75cm, draw] at (-2, 4) {$\nu_1$};
		\node (h') [rectangle, minimum height=1.5cm, minimum width=.75cm, draw] at (-2, 2) {$\nu_2$};
		\node (f) [rectangle, minimum height=1.5cm, minimum width=.75cm, draw] at (2, 4) {$\lambda_1$};
		\node (f') [rectangle, minimum height=1.5cm, minimum width=.75cm, draw] at (2, 2) {$\lambda_2$};
		\node (g') [rectangle, minimum height=1.5cm, minimum width=.75cm, draw] at (0, -1) {$\mu_2$};
		\node (dummy1) at (0, 3.5) {}; \node (dummy2) at (0, 2.5) {};
		\draw [-] (Y) to (h); \draw [-] (h) to [out=0, in=180] (c.west |- dummy1); \draw [-] (c.east |- dummy1) to [out=0, in=180] (f); \draw [-] (f) to (X);
		\draw [-] (h') to [out=0, in=180] (c.west |- dummy2); \draw [-] (c.east |- dummy2) to [out=0, in=180] (f');
		\draw [-] (f') to [out=0, in=90] (3, 1.5) to [out=-90, in=90] (-1, -.5) to [out=-90, in=180] (g');
		\draw [-] (g') to [out=0, in=-90] (1, -.5) to [out=90, in=-90] (-3, 1.5) to [out=90, in=180] (h');
	\end{tikzpicture} \end{center}
	By the coherence theorem for traced monoidal categories, these denote equal morphisms.
\end{proof}

\begin{definition}
	Let $\C$ be a symmetric monoidal category with a context $\overline\C$, and let $X, Y$ be objects of $\C$.
	An \emph{open game} $\G : X \to Y$ over $\C$ consists of
	\begin{itemize}
		\item A set $\Sigma_\G$ of strategy profiles
		\item A labelling function $\G_- : \Sigma_\G \to \C (X, Y)$
		\item A winning condition $|\G| \subseteq \Sigma_\G \times \overline\C (X, Y)$
	\end{itemize}
\end{definition}

Given open games $\G : X \to Y$ and $\H : Y \to Z$ over $\C$, their sequential composition $\H \circ \G : X \to Z$ is defined by $\Sigma_{\H \circ \G} = \Sigma_\G \times \Sigma_\H$, $(\H \circ \G)_{\sigma, \tau} = \H_\tau \circ \G_\sigma$ and
\[ |\H \circ \G|^{\sigma, \tau}_\kappa \iff |\G|^\sigma_{\overline\C (\id_X, \H_\tau) (\kappa)} \wedge |\H|^\tau_{\overline\C (\G_\sigma, \id_Z) (\kappa)} \]
Given open games $\G : X_1 \to Y_1$ and $\H : X_2 \to Y_2$ over $\C$, their simultaneous composition $\G \otimes \H : X_1 \otimes X_2 \to Y_1 \otimes Y_2$ is defined by $\Sigma_{\G \otimes \H} = \Sigma_\G \times \Sigma_\H$, $(\G \otimes \H)_{\sigma, \tau} = \G_\sigma \otimes \H_\tau$ and 
\[ |\G \otimes \H|^{\sigma, \tau}_\kappa \iff |\G|^\sigma_{\H_\tau / \kappa} \wedge |\H|^\tau_{\G_\sigma \setminus \kappa} \]

\begin{proposition}
	For any symmetric monoidal category $\C$ with a context, there is a symmetric monoidal category $\OG (\C)$ of open games over $\C$.
	When $\C = \Lens (\Set)$ with the usual context, we obtain the original category of open games.
\end{proposition}

The proof of this proposition formally follows the proof that $\OG$ is a symmetric monoidal category.
(The clearest presentation is in \cite[section 5 \& appendix]{hedges_morphisms_open_games}.)
The definition of a context contains precisely the conditions needed for this proof to work.

Given a morphism $\lambda : X \to Y$ of $\C$, we define an open game $\lambda : X \to Y$ over $\C$ by $\Sigma_{\lambda} = \{ * \}$, $\lambda_{*} = \lambda$ and $|\lambda|^*_\kappa$ holding for all $\kappa$.

\begin{proposition}
	This defines a faithful identity-on-objects symmetric monoidal functor $\C \to \OG (\C)$.
\end{proposition}

\section{Morphisms of contexts}

Given a pair of categories with contexts $\C$, $\D$, it is possible to relate open games in $\OG (\C)$ to open games in $\OG (\D)$ if we have a strict monoidal functor $\C \to \D$ that is compatible with the context functors (despite the fact that $\OG (-)$ is not functorial due to mixed variance).
In this section we will prove (mostly for completeness of the presentation) that when $\C$ is traced cartesian, the functor $-^* : \Lens (\C) \to \Int (\C)$ satisfies the required properties.
This allows us to compare open games over $\Lens (\DCPO)$ and $\Int (\DCPO)$ for example.

The reason we do not develop this idea fully is that it does not seem possible to obtain a strict (or even strong) monoidal functor $\Lens (\Set) \to \Lens (\DCPO)$ that would allow us to understand `computable game theory' as far as possible in terms of classical game theory.
Ultimately this stems from the lack of a suitable product-preserving functor $\Set \to \DCPO$.

It does seem possible to overcome this using machinery that is known.
One possibility is to consider $\DCPO$ with the smash product, which is a monoidal product that is not the categorical product, and then consider optics over this.
Open games over a particular category of optics (over the monoidal category of conditional probability distributions) were considered in the context of Bayesian games \cite{bolt_hedges_zahn_bayesian_open_games}, but they are more subtle and less intuitive so we leave this for future work.

\begin{definition}\label{def:context-morphism}
	Let $\C$ and $\D$ be symmetric monoidal categories with contexts $\overline\C$ and $\overline\D$.
	A strict morphism of contexts is a strict symmetric monoidal functor $F : \C \to \D$ together with a monoidally natural family of functions
	\[ \overline F (X, Y) : \overline\C (X, Y) \to \overline\D (F (X), F (Y)) \]
	such that
	\[ \begin{tikzcd}
		\overline\C (X_1 \otimes X_2, Y_1 \otimes Y_2) \ar[rrr, "f / -"] \ar[ddd, "{\overline F (X_1 \otimes X_2, Y_1 \otimes Y_2)}"] &&& \overline\C (X_1, Y_1) \ar[ddd, "{\overline F (X_1, Y_1)}"] \\ \\ \\
		\overline\D (F (X_1 \otimes X_2), F (Y_1 \otimes Y_2)) \ar[rrr, "F (f) / -"] &&& \overline\D (F (X_1), F (Y_1))
	\end{tikzcd} \]
	commutes for all $f : X_2 \to Y_2$.
\end{definition}

Defining non-strict morphisms of contexts takes a bit more care, but is not necessary for our purposes.

Given a traced monoidal category $\C$, the category $\Int (\C)$ is compact closed, and hence in particular traced monoidal.
We consider it to have the context defined for traced monoidal categories.
That is,
\[ \overline{\Int (\C)} \left( \diset X S, \diset Y R \right) = \Int (\C) \left( \diset Y R, \diset X S \right) = \C (Y \otimes S, X \otimes R) \]

\begin{proposition}
	Let $\C$ be a traced cartesian category.
	Then $-^* : \Lens (\C) \to \Int (\C)$ can be made into a strict morphism of contexts, by defining
	\[ \overline{\ *\ } : \C (1, X) \times \C (Y, R) \to \C (Y \times S, X \times R) \]
	to take $(h, k)$ to
	
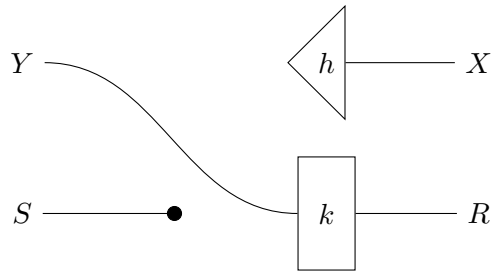
\begin{center} \begin{tikzpicture}
		\node (Y) at (0, 2) {$Y$}; \node (S) at (0, 0) {$S$}; \node (X) at (6, 2) {$X$}; \node (R) at (6, 0) {$R$};
		\node (node) [circle, scale=.5, fill=black, draw] at (2, 0) {};
		\node (h) [isosceles triangle, isosceles triangle apex angle=90, shape border rotate=180, minimum width=1.5cm, draw] at (4, 2) {$h$};
		\node (k) [rectangle, minimum height=1.5cm, minimum width=.75cm, draw] at (4, 0) {$k$};
		\draw [-] (Y) to [out=0, in=180] (k); \draw [-] (S) to (node); \draw [-] (h) to (X); \draw [-] (k) to (R);
	\end{tikzpicture} \end{center}
\end{proposition}

\begin{proof}
	We already checked that $-^*$ is strict symmetric monoidal.
	We get naturality for free by noting that $\overline{\ *\ }$ can be equivalently defined by
	\begin{align*}
		\C (I, X) \times \C (Y, R) &\overset\cong\longrightarrow \Lens (\C) \left( I, \diset X S \right) \times \Lens (\C) \left( \diset Y R, I \right) \\
		&\overset\circ\longrightarrow \Lens (\C) \left( \diset Y R, \diset X S \right) \\
		&\overset{-^*}\longrightarrow \Int (\C) \left( \diset Y R, \diset X S \right)
	\end{align*}
	
	Suppose we have a $P$-strategy $\lambda : \diset{X_2}{S_2} \to \diset{Y_2}{R_2}$.
	We must verify that the square
	\[ \begin{tikzcd}
		\C (1, X_1 \times X_2) \times \C (Y_1 \times Y_2, R_2 \times R_1) \ar[r, "\lambda / -"] \ar[dd] & \C (1, X_1) \times \C (Y_1, R_1) \ar[dd] \\ \\ 
		\C (Y_1 \times Y_2 \times S_2 \times S_1, X_1 \times X_2 \times R_2 \times R_1) \ar[r, "\lambda^* / -"] & \C (Y_1 \times S_1, X_1 \times R_1)
	\end{tikzcd} \]
	commutes.
	Chasing a context $(h, k)$ around the top yields
	
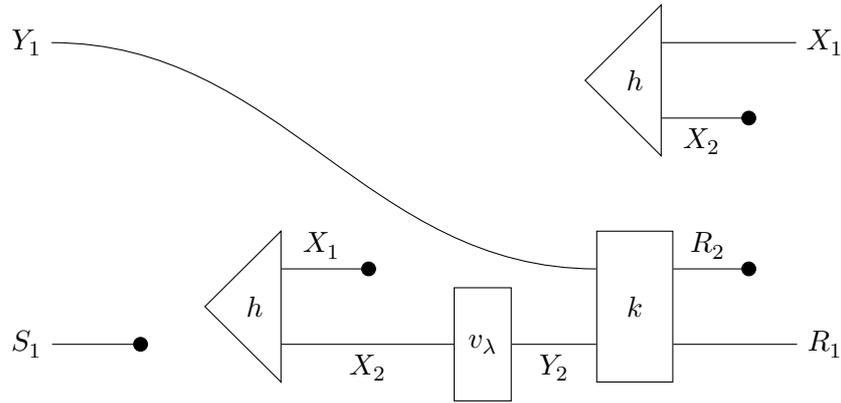
\begin{center} \begin{tikzpicture}
		\node (Y) at (-.5, 4) {$Y_1$}; \node (S) at (-.5, 0) {$S_1$}; \node (X) at (10, 4) {$X_1$}; \node (R) at (10, 0) {$R_1$};
		\node (h1) [isosceles triangle, isosceles triangle apex angle=90, shape border rotate=180, minimum width=2cm, draw] at (7.5, 3.5) {$h$};
		\node (h2) [isosceles triangle, isosceles triangle apex angle=90, shape border rotate=180, minimum width=2cm, draw] at (2.5, .5) {$h$};
		\node (v) [rectangle, minimum height=1.5cm, minimum width=.75cm, draw] at (5.5, 0) {$v_\lambda$};
		\node (k) [rectangle, minimum height=2cm, minimum width=1cm, draw] at (7.5, .5) {$k$};
		\node (m1) [circle, scale=.5, fill=black, draw] at (9, 3) {}; \node (m3) [circle, scale=.5, fill=black, draw] at (9, 1) {};
		\node (m2) [circle, scale=.5, fill=black, draw] at (1, 0) {}; \node (m4) [circle, scale=.5, fill=black, draw] at (4, 1) {};
		\draw [-] (Y) to [out=0, in=180] (k.west |- m3); \draw [-] (h1.east |- X) to (X); \draw [-] (h1.east |- m1) to node [below] {$X_2$} (m1);
		\draw [-] (k.east |- m3) to node [above] {$R_2$} (m3); \draw [-] (k.east |- R) to (R);
		\draw [-] (S) to (m2); \draw [-] (h2.east |- m4) to node [above] {$X_1$} (m4); \draw [-] (h2.east |- v) to node [below] {$X_2$} (v);
		\draw [-] (v) to node [below] {$Y_2$} (k.west |- v);
	\end{tikzpicture} \end{center}
	As a useful intermediate point, since $h$ is a comonoid homomorphism this is equivalent to
	
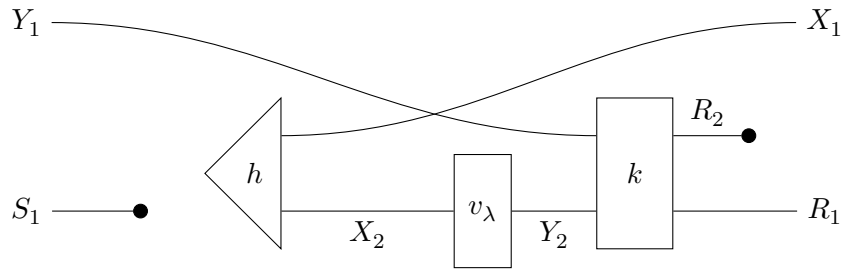
\begin{center} \begin{tikzpicture}
		\node (Y) at (-.5, 2.5) {$Y_1$}; \node (S) at (-.5, 0) {$S_1$}; \node (X) at (10, 2.5) {$X_1$}; \node (R) at (10, 0) {$R_1$};
		\node (h) [isosceles triangle, isosceles triangle apex angle=90, shape border rotate=180, minimum width=2cm, draw] at (2.5, .5) {$h$};
		\node (v) [rectangle, minimum height=1.5cm, minimum width=.75cm, draw] at (5.5, 0) {$v_\lambda$};
		\node (k) [rectangle, minimum height=2cm, minimum width=1cm, draw] at (7.5, .5) {$k$};
		\node (m3) [circle, scale=.5, fill=black, draw] at (9, 1) {}; \node (m2) [circle, scale=.5, fill=black, draw] at (1, 0) {};
		\draw [-] (Y) to [out=0, in=180] (k.west |- m3);
		\draw [-] (k.east |- m3) to node [above] {$R_2$} (m3); \draw [-] (k.east |- R) to (R);
		\draw [-] (S) to (m2); 
		\draw [-] (v) to node [below] {$Y_2$} (k.west |- v);
		\draw [-] (h.east |- m3) to [out=0, in=180] (X); \draw [-] (h.east |- v) to node [below] {$X_2$} (v);
	\end{tikzpicture} \end{center}
	On the other hand, chasing $(h, k)$ around the bottom yields
	
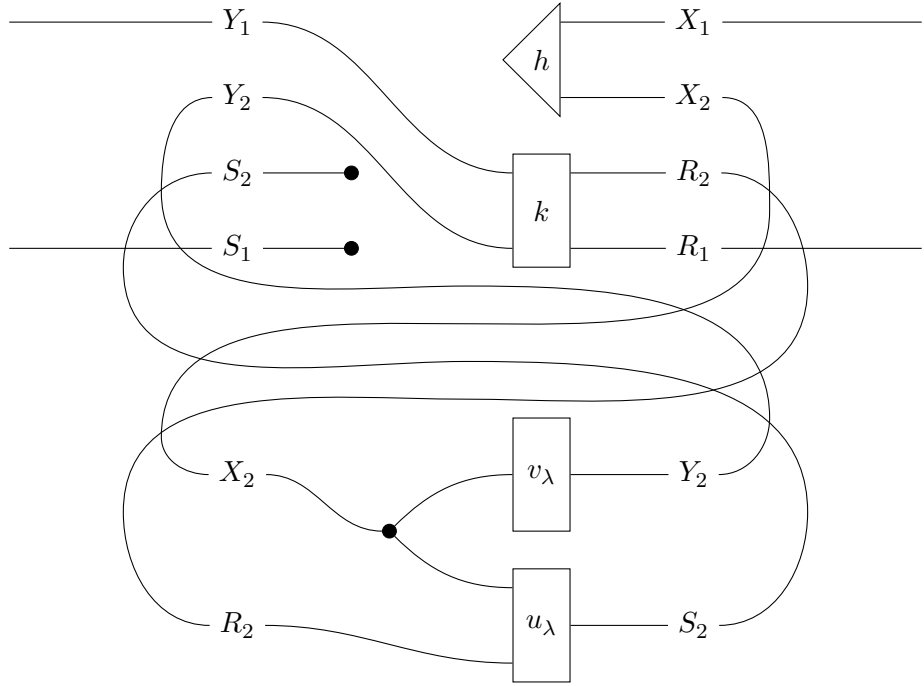
\begin{center} \begin{tikzpicture}
		\node (Y1) at (0, 8) {$Y_1$}; \node (Y2) at (0, 7) {$Y_2$}; \node (S2) at (0, 6) {$S_2$}; \node (S1) at (0, 5) {$S_1$};
		\node (X22) at (0, 2) {$X_2$}; \node (R22) at (0, 0) {$R_2$};
		\node (X1) at (6, 8) {$X_1$}; \node (X2) at (6, 7) {$X_2$}; \node (R2) at (6, 6) {$R_2$}; \node (R1) at (6, 5) {$R_1$};
		\node (Y22) at (6, 2) {$Y_2$}; \node (S22) at (6, 0) {$S_2$};
		\node (h) [isosceles triangle, isosceles triangle apex angle=90, shape border rotate=180, minimum width=1.5cm, draw] at (4, 7.5) {$h$};
		\node (k) [rectangle, minimum height=1.5cm, minimum width=.75cm, draw] at (4, 5.5) {$k$};
		\node (m1) [circle, scale=.5, fill=black, draw] at (1.5, 6) {}; \node (m2) [circle, scale=.5, fill=black, draw] at (1.5, 5) {};
		\node (m3) [circle, scale=.5, fill=black, draw] at (2, 1.25) {};
		\node (v) [rectangle, minimum height=1.5cm, minimum width=.75cm, draw] at (4, 2) {$v_\lambda$};
		\node (u) [rectangle, minimum height=1.5cm, minimum width=.75cm, draw] at (4, 0) {$u_\lambda$};
		\node (dummy1) at (0, .5) {}; \node (dummy2) at (0, -.5) {};
		\draw [-] (Y1) to [out=0, in=180] (k.west |- S2); \draw [-] (Y2) to [out=0, in=180] (k.west |- S1); \draw [-] (S2) to (m1); \draw [-] (S1) to (m2);
		\draw [-] (h.east |- X1) to (X1); \draw [-] (h.east |- X2) to (X2); \draw [-] (k.east |- R2) to (R2); \draw [-] (k.east |- R1) to (R1);
		\draw [-] (X22) to [out=0, in=180] (m3); \draw [-] (m3) to [out=45, in=180] (v); \draw [-] (v) to (Y22);
		\draw [-] (m3) to [out=-45, in=180] (u.west |- dummy1); \draw [-] (R22) to [out=0, in=180] (u.west |- dummy2); \draw [-] (u) to (S22);
		\draw [-] (-3, 8) to (Y1); \draw [-] (-3, 5) to (S1); \draw [-] (X1) to (9, 8); \draw [-] (R1) to (9, 5);
		\draw [-] (Y2) to [out=180, in=90] (-1, 5.75) to [out=-90, in=180] (3, 4.5) to [out=0, in=90] (7, 2.75) to [out=-90, in=0] (Y22);
		\draw [-] (X2) to [out=0, in=90] (7, 5.5) to [out=-90, in=0] (3, 4) to [out=180, in=90] (-1, 2.5) to [out=-90, in=180] (X22);
		\draw [-] (S2) to [out=180, in=90] (-1.5, 4.75) to [out=-90, in=180] (3, 3.5) to [out=0, in=90] (7.5, 1.5) to [out=-90, in=0] (S22);
		\draw [-] (R2) to [out=0, in=90] (7.5, 4.5) to [out=-90, in=0] (3, 3) to [out=180, in=90] (-1.5, 1.5) to [out=-90, in=180] (R22);
	\end{tikzpicture} \end{center}
	To see the equivalence of this diagram to the previous one modulo traced cartesian categories, trace the deletion on $S_2$ backwards.
\end{proof}

\section{Compositional computable game theory}

We can finally put all the pieces together, by considering the category $\OG (\Int (\DCPO))$.
Concretely, for DCPOs $X, S, Y, R$, such an open game $\G : \diset X S \to \diset Y R$ consists of:
\begin{enumerate}
	\item A set $\Sigma$ of \emph{strategy profiles}
	\item A family of continuous \emph{play functions} $\mathbf P_\G (\sigma) : X \times R \to Y$
	\item A family of continuous \emph{coplay functions} $\mathbf C_\G (\sigma) : X \times R \to S$
	\item An equilibrium set $\mathbf E_\G (h, k) \subseteq \Sigma$ for each continuous \emph{history} $h : Y \times S \to X$ and \emph{continuation} $k : Y \times S \to R$
\end{enumerate}
This definition is written in the style of the original concrete definition of open games in \cite{hedges_etal_compositional_game_theory} for ease of comparison.
Specifically, besides changing the base category from $\Set$ to $\DCPO$ this definition differs by making $\mathbf C_\G$ additionally a function of $R$, $h$ a function of $Y$ and $S$, and $k$ a function of $S$.

The category $\OG (\Int (\DCPO))$ is compact closed, as a result of applying proposition \ref{prop:ccc-transport} to the zero-player functor $\Int (\DCPO) \to \OG (\Int (\DCPO))$.

As an exercise, we work out the transpose $\mathcal G^* : \diset R Y \to \diset S X$ of a general open game $\mathcal G : \diset X S \to \diset Y R$ over $\Int (\C)$.
The set of strategy profiles stays the same up to isomorphism, $\Sigma_{\G^*} \cong \Sigma_\G$, because the transpose is defined by composition with various open games whose set of strategy profiles is $1$.
The play function is modified by taking the transpose in $\Int (\C)$, which in the end simply exchanges the play and coplay functions $X \times R \to Y$, $X \times R \to S$ and swaps their inputs.
A context $(h, k)$, for $h : S \times Y \to R$ and $k : S \times Y \to X$ is again swapped to give a context $h' = k : Y \times S \to X$ and $k' = h : Y \times S \to R$ for $\mathcal G$, so equilibrium is defined by $|\mathcal G^*|^\sigma_{h, k} \iff |\mathcal G|^\sigma_{k, h}$.

Given a continuous function $f : X \to Y$, the covariant and contravariant liftings $f : \diset X 1 \to \diset Y 1$ and $f^* : \diset 1 Y \to \diset 1 X$ are now transposes of each other.
Thus we are conservatively extending the notion of duality that already exists in categories of open games.

Recall that for a decision $\D_{Y | X} : \diset X 1 \to \diset{Y}{\mathbb R}$ over $\Lens (\Set)$, a context is a pair $h : X$ and $k : Y \to \mathbb R$, and the equilibrium condition for a strategy $\sigma : X \to Y$ is that $\sigma (h) \in \arg\max (k)$.
We can make a similar definition over $\Lens (\DCPO)$ given a suitable domain of reals $\mathbb R$ and a suitable $\arg\max$ operator defined on continuous functions $Y \to \mathbb R$.
There are several options for defining these, and we remain largely agnostic between them.
(We do not assume that $\arg\max$ is internalised in $\DCPO$ as a function $\mathbb R^Y \to \mathcal P (Y)$ for some powerdomain $\mathcal P$, since a naive definition would not be continuous.)


A simple example of a domain of reals that can serve as a mental model is the domain of closed intervals $[x, y]$ with the reverse inclusion order, together with $\bot_\mathbb R = (-\infty, +\infty)$.
Here $[x, y]$ represents an approximation of some $z \in [x, y]$, and a standard real number $z$ is represented by the degenerate interval $[z, z]$.
Note that the $\arg\max$ operator is still defined for the standard order on reals (which must be extended to all elements of the domain), which is not related to the inclusion order.
As a minimal requirement in order to work out an example later, we suppose that $\bot_\mathbb R$ is below every standard real in the extended standard order.
This corresponds to the assumption that players always prefer a terminating payoff, no matter how small, to a nonterminating one.


Over $\Int (\DCPO)$, the context for a decision $\D_{X, Y}$ has the form
\[ \kappa = (h, k) : \overline{\Int (\DCPO)} \left( \diset X 1, \diset{Y}{\mathbb R} \right) \cong \DCPO (Y, X) \times \DCPO (Y, \mathbb R) \]
Notice that since the coutility type $S = 1$ is the terminal DCPO, the only difference from a context over $\Lens (\C)$ is that the history $h$ may depend on the move from $Y$.
That is, the future action may affect the past observation.

We define the decision $\D = \D_{Y | X} : \diset X 1 \to \diset{Y}{\mathbb R}$ over $\Int (\DCPO)$ to have $\Sigma_\D = \DCPO (X, Y)$, and the play function
\[ \DCPO (X, Y) \to \Int (\DCPO) \left( \diset X 1, \diset{Y}{\mathbb R} \right) \cong \DCPO (X \times \mathbb R, Y) \]
given by composition with the projection.
A natural definition for equilibrium is that the least fixpoint $y$ of $y = \sigma (h (y))$ is in $\arg\max (k)$, which we write
\[ \left| \D \right|^\sigma_{h, k} \iff \mu y . \sigma (h (y)) \in \arg\max (k) \]

%


As a worked example, we can build a 2-player game in which each player's strategy may be contingent on the choice of the other, something that is causally absurd.
Let the game $\mathcal G : I \to I$ be defined by the string diagram
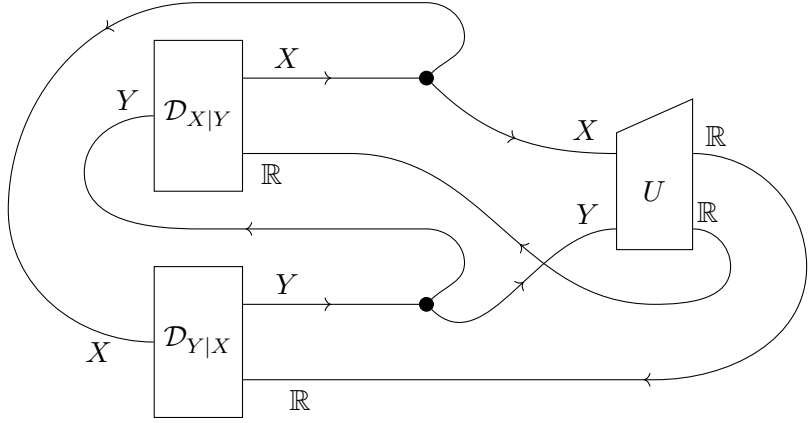
\begin{center} \begin{tikzpicture}[auto]
	\node (DYX) [rectangle, minimum height=2cm, draw] at (0, 3) {$\mathcal D_{X | Y}$};
	\node (DXY) [rectangle, minimum height=2cm, draw] at (0, 0) {$\mathcal D_{Y | X}$};
	\node (copy1) [circle, scale=.5, fill=black, draw] at (3, 3.5) {}; \node (copy2) [circle, scale=.5, fill=black, draw] at (3, .5) {};
	\node (dummy1) at (0, 2.5) {}; \node (dummy2) at (0, 1.5) {}; \node (dummy3) at (0, 1) {}; \node (dummy4) at (0, -.5) {};
	\node [trapezium, trapezium left angle=0, trapezium right angle=75, shape border rotate=90, trapezium stretches=true, minimum height=1cm, minimum width=2cm, draw] (U) at (6, 2) {$U$};
	\draw [->-] (DYX.east |- copy1) to node [near start] {$X$} (copy1); \draw [->-] (DXY.east |- copy2) to node [near start] {$Y$} (copy2);
	\draw [->-] (copy1) to [out=45, in=-90] (3.5, 4.05) to [out=90, in=0] (3, 4.5) to (0, 4.5) to [out=180, in=90] (-2.5, 1.75) to [out=-90, in=180] node [below, near end] {$X$} (DXY.west);
	\draw [->-] (copy2) to [out=45, in=-90] (3.5, 1.05) to [out=90, in=0] (3, 1.5) to (0, 1.5) to [out=180, in=-90] (-1.5, 2.25) to [out=90, in=180] node [above, near end] {$Y$} (DYX.west);
	\draw [->-] (copy1) to [out=-45, in=180] node [near end] {$X$} (U.west |- dummy1);
	\draw [->-] (copy2) to [out=-45, in=180] node [above, very near end] {$Y$} (U.west |- dummy2);
	\draw [->-] (U.east |- dummy1) to [out=0, in=90] node [above, very near start] {$\R$} (8, 1) to [out=-90, in=0] (6, -.5) to [out=180, in=0] node [below, very near end] {$\R$} (DXY.east |- dummy4);
	\draw [->-] (U.east |- dummy2) to [out=0, in=90] node [above, near start] {$\R$} (7, 1) to [out=-90, in=0] (6, .5) to [out=180, in=0] (2, 2.5) to [out=180, in=0] node [below, near end] {$\R$} (DYX.east |- dummy1);
\end{tikzpicture} \end{center}
We suppose $X$ and $Y$ to be finite flat domains, say $X = Y = \{ \bot, a, b \}$.
We also suppose that $U$ is zero-player and encodes some function $U : X \times Y \to \mathbb R \times \mathbb R$, where $\mathbb R$ in the latter is the set of standard real numbers.
That is to say, every $U (x, y)$ is some pair of total real numbers for $x, y \neq \bot$.
As a specific example, let $U$ be the payoff matrix of matching pennies, extended as follows:
\[ U (x, y) = \begin{cases}
	(1, 0) &\text{ if } x = y \neq \bot \\
	(0, 1) &\text{ if } x \neq y, x \neq \bot, y \neq \bot \\
	(\bot_\mathbb R, \bot_\mathbb R) &\text{ if } x = \bot \text{ or } y = \bot
\end{cases} \]
Matching pennies is an interesting example here because it exhibits second-move advantage: either player would benefit from the ability to play contingently on the other's move.

The set of strategy profiles of this game is $\Sigma = \DCPO (X, Y) \times \DCPO (Y, X)$.
The equilibrium condition $|\G|^{\sigma, \tau}$ holds iff
\[ \mu x . \sigma (\tau (x)) \in \arg\max_x U_1 (x, \tau (x)) \]
\[ \mu y . \tau (\sigma (y)) \in \arg\max_y U_2 (\sigma (y), y) \]

We can now check these conditions on some specific examples.
First suppose that $\sigma$ and $\tau$ are both constant functions, say $\sigma (y) = a$ and $\tau (x) = b$, including for $y = \bot$ and $x = \bot$.
(These are `lazy' functions: they terminate with a total value even when their input does not.)
We can then directly calculate that $\mu x . \sigma (\tau (x)) = a$ and $\mu y . \tau (\sigma (y)) = b$.
The first player has incentive to deviate because $a \not\in \arg\max_x U_1 (x, b) = \{ b \}$, although the second player is satisfied since $b \in \arg\max_y U_2 (a, y) = \{ b \}$.
Thus $(\sigma, \tau)$ is not an equilibrium of this game.
By this reasoning $\mathcal G$ has no `lazy' equilibria of this form, since matching pennies has no pure strategy Nash equilibria.

Next consider the strategies $\sigma (y) = y$ and $\tau (x) = a$, in which player 1 seizes the second-move advantage by playing the optimal response to player 2's move, namely copying it.
Then $\mu x . \sigma (\tau (x)) = \mu y . \tau (\sigma (y)) = a$, and $(\sigma, \tau)$ is an equilibrium since $\arg\max_x U_1 (x, a) = \{ a \}$ and $\arg\max_y U_2 (y, y) = \{ a, b \}$.
There is another equilibrium given by $\sigma (y) = y$ and $\tau (x) = b$.
Notice that player 2 could play $\bot$ and deadlock the play, but we assume that she prefers the `losing' total payoff of $0$.
Similarly there are two more equilibria in which it is player 2 who takes the second-move advantage with $\tau (x) = \overline x$, given by $\overline a = b$ and $\overline b = a$ (and in consequence, $\overline \bot = \bot$).

Finally, suppose that both players attempt to move second, with the strategy profile $\sigma (y) = y$ and $\tau (x) = \overline x$.
Then $\mu x . \sigma (\tau (x)) = \mu y . \tau (\sigma (y)) = \bot$: the play deadlocks as each player waits for the other to move first.
However $\arg\max_x U_1 (x, \overline x) = \arg\max_y U_2 (y, y) = \{ a, b \}$, so both players have incentive to deviate and $(\sigma, \tau)$ is not an equilibrium.
Given our assumptions, either player would prefer to move first and take the losing total payoff, rather than deadlocking the play.


\bibliographystyle{alpha}
\bibliography{\string~/Dropbox/Work/refs}

\end{document}